\documentclass[11pt,a4paper]{article}
\usepackage{amsmath,amssymb,amsfonts,amsbsy,amsthm,epsfig,graphicx,rotating}
\usepackage[left=7.5pc,right=7.5pc,top=1in,bottom=1in]{geometry}
\usepackage{epstopdf}
\usepackage{subfigure}
\usepackage{array,multirow}
\usepackage{graphicx}
\usepackage{adjustbox}

\theoremstyle{plain}
\newtheorem{problem}{Problem}[section]
\newtheorem{theorem}{Theorem}[section]

\newtheorem{lemma}[theorem]{Lemma}

\theoremstyle{remark}

\theoremstyle{definition}

\begin{document}



%
\def\spreme{{\sc{spreme}}}
\def\beq{\begin{equation}}
\def\eeq{\end{equation}}
\def\bea{\begin{eqnarray}}
\def\eea{\end{eqnarray}}
\def\A{{\mathbf{A}}}
\def\bsigma{{\mbox{\boldmath$\sigma$}}}
\def\bep{{\mbox{\boldmath$\epsilon$}}}
\def\ep{{\mbox{$\epsilon$}}}
\def\bxi{{\mbox{\boldmath$\xi$}}}
\def\u{{\mbox{\boldmath$u$}}}
\def\v{{\mbox{\boldmath$v$}}}
\def\d{{\mbox{\boldmath$d$}}}
\def\t{{\mbox{\boldmath$t$}}}
\def\T{{\mathbf{T}}}
\def\I{{\mathbf{I}}}
\def\B{{\mathbf{B}}}
\def\x{{\mbox{\boldmath$x$}}}
\def\y{{\mbox{\boldmath$y$}}}
\def\w{{\mbox{\boldmath$w$}}}
\def\K{{\mbox{\boldmath$K$}}}
\def\f{{\mbox{\boldmath$f$}}}
\def\matlab{{\sc{matlab}}}
\def\spreme{{\sc{spreme}}}
\def\tr{\mbox{tr}}



\title{\textit{
Recovering vector displacement estimates in quasistatic elastography using sparse 
relaxation of the momentum equation}}

\author{
{
Olalekan A.\ Babaniyi\textsuperscript{a}, 
Assad A.\ Oberai\textsuperscript{b}, 
Paul E.\ Barbone\textsuperscript{a}$^{\ast}$\thanks{$^\ast$Corresponding author. Email: barbone@bu.edu}
} \\
\textsuperscript{a}Mechanical Engineering, Boston University, Boston, MA, USA;\\
\textsuperscript{b}Mechanical Aerospace and Nuclear Engineering, RPI, Troy, NY, USA
}

\maketitle

\begin{abstract}
We consider the problem of estimating the $2D$ vector displacement field in a heterogeneous 
elastic solid deforming under plane stress conditions.  The problem is motivated by applications 
in quasistatic elastography.  
From precise and accurate measurements of one component of the $2D$ vector displacement field and very limited information of the second component, the method reconstructs the second component quite accurately. 
No a priori knowledge of the heterogeneous distribution of material properties is required. This method relies on using a special form of the momentum equations to filter ultrasound displacement measurements to produce more precise estimates.  
We verify the method with applications to simulated displacement data.  We validate 
the method with applications to displacement data measured from a tissue mimicking phantom, and in-vivo data;  significant improvements are noticed in the filtered displacements recovered from all the tests.  In verification studies, error in lateral displacement estimates decreased from about $50\%$ to about $2\%$, and strain error decreased from more than $250\%$ to below $2\%$.  
\end{abstract}

\section{Introduction}
Ultrasound elasticity imaging (UEI), or elastography, is a rapidly growing field, primarily due to the increasing interest in the non-invasive quantification of the mechanical properties of soft tissues
\cite{gao1996imaging,ophir1999elastography,parker2005unified,greenleaf2003selected,parker2011imaging,doyley2012model}. A crucial step in UEI is the estimation of a tissue's motion. This motion is typically generated by a quasistatic compression, and the deformation of the tissue is measured with ultrasound. The measured deformation and an appropriate mathematical model can then be used to infer the mechanical properties of the soft tissue. 

Ultrasound measurements of tissue deformation are typically much more precise in the ``axial" 
direction (i.e.\ the direction of the ultrasound beam), than in the so-called ``lateral" direction, the image direction transverse to the direction of the ultrasound beam.  
The origin of the difference in precision is due to the anisotropy of the ultrasound point spread function (PSF).  The width of the peak of the radio frequency (RF) PSF, in the axial direction (the direction of sound propagation) 
is roughly an order of magnitude smaller than that in the lateral direction \cite{book:tlszabo}.

With only one component of the displacement field measured precisely, it is impossible to compute precise estimates of the full strain tensor field, which would otherwise be useful information to researchers who do strain imaging.   Furthermore, one goal of measuring the displacement field is as input to an inverse problem to estimate material property distributions \cite{barbone2010review,doyley2012model}.   
Computing the material properties from only one component of the displacement field is computationally intensive \cite{barbone2010review,doyley2012model}.  In that situation, an iterative inversion approach is required, which may takes hundreds to thousands of iterations to converge.  Such an iterative inversion code usually needs displacement boundary conditions to drive the forward problem, and using the noisy measurements as boundary conditions can corrupt the material property estimates \cite{jDordLNLinear,goenezen2012linear,richards2009quantitative}. 
On the other hand, efficient direct methods to compute material properties from full vector displacement field measurements exist 
\cite{
albocher2009adjoint, 
barbone2010adjoint, 
albocher2013approaches}.  
The present contribution uses an inverse problem approach that aims to improve lateral displacement estimates from ultrasound measurements. 

Much research has been performed in order to recover better lateral displacements both in the ultrasound elasticity imaging and the ultrasound blood flow imaging field. There are four main approaches that dominate the literature. In \cite{konofagou1998new,geiman2000novel,konofagou2000precision,chen2004lateral,ebbini2006phase,brusseau20082,deprez20093d,kim2011autocorrelation,rivaz2011real,mailloux1987computer,basarab2008method} the authors use various forms of interpolation within standard or modified block-matching algorithms to estimate subsample lateral displacements. These approaches try to make the best use of available data as is, but do not overcome the fundamental difficulty posed by the anisotropic PSF.
An alternative approach is to use beam steering to measure displacements at arbitrary angles, and then use the arbitrary angled displacement estimates to determine the axial and lateral displacements as described in \cite{tanter2002ultrafast,techavipoo2004estimation,rao2007normal,hansen2010full,xu2013normal,abeysekera2012analysis}. This method uses additional image data to overcome the limitations of a single PSF. 
A third approach, pursued by \cite{korukonda2011estimating,korukonda2012visualizing,jensen1998new,liebgott2007psf}, is to use novel beam forming techniques to modulate the RF signal in the lateral direction so that it contains more phase information. This approach tries to reduce the anisotropy of the PSF.
A fourth approach, pursued by \cite{richards2009quantitative,lubinski1996lateral,skovoroda1998nonlinear,o2001strain}, is to constrain the displacements to satisfy the incompressibility condition. Many different types of soft tissue can be assumed to be incompressible because they are mostly composed of water.

Among the existing approaches, the incompressibility processing method introduced in \cite{lubinski1996lateral} is arguably the most similar in spirit to the method presented here to reconstruct better lateral displacements. In that paper, the authors assumed that the deformation is linear, incompressible, and plane strain in character. With these assumptions, they formulated a minimization problem where the objective was to find lateral displacements that minimized the error between the measurements and predicted lateral displacements that satisfied the modeling assumptions. The incompressibility processing described in \cite{o2001strain} is basically the same as the method developed in \cite{lubinski1996lateral}, with the addition of weighting functions to the lateral displacement reconstruction equations. These weighting functions are based on the absolute value and the gradient of the cross-correlation function. In \cite{skovoroda1998nonlinear}, the method developed in \cite{lubinski1996lateral} is extended to accommodate large plane strain deformations. Finally, in \cite{richards2009quantitative}, the incompressiblility constraint is used to improve the lateral and elevational displacements for a 3D deformation. 

In this paper, we introduce a fundamentally new approach to evaluating lateral displacements.  We will be using a spatial regularization term within an inverse problem formulation to adaptively enforce, and locally relax, a special form of the momentum conservation equation on the measured displacement field. 
Special weighting functions will be used to place more emphasis on the axial component of the displacement field.


The rest of the paper is arranged as follows: section \ref{sec:fomu} describes the new formulation to estimate better lateral displacements; we refer to this method as ``{\spreme},'' which stands for SParse RElaxation of the Momentum Equation. Section \ref{sec:veri} describes the simulation experiments used to test and verify the displacement filtering algorithm. Section \ref{sec:vald} describes the phantom experiments used to validate the performance of the displacement filtering algorithm on ultrasound measured data. Section \ref{sec:app} describes the in-vivo experiments used to further test the performance of the displacement filtering algorithm on patient data. Finally, a brief summary and conclusion is given in section \ref{sec:summConc} along with possible future directions of the work. The appendix describes a simpler formulation that produces apparently improved displacement and strain fields. Despite appearances, however, the resulting strain and displacement fields are non-physical.

\section{Formulation} 
\label{sec:fomu}
We suppose we are given measured displacements in a 2D region $\Omega$. We further suppose the measurements of the axial displacement, $u_y$, are significantly more precise and accurate than the measurements of the lateral displacement. We make use of the assumption that within the observed plane, the body's deformation may be well approximated by the plane stress approximation. These assumptions lead to the following problem statement:

\begin{problem}
Given 2D measured displacement $\u_{m}(x,y)$, for all $\x\in\Omega$, find the 2D displacement vector $\u(x,y)$, and the 2D linearized strain tensor $\bep(x,y)$ such that:
\bea
\pi_{o}[\u] &=& \frac{1}{2}||(\u_{m} - \u)||^{2}_{_{\mbox{T}}} \nonumber \\
 &=& \frac{1}{2} \int_{\Omega} (\u_{m} - \u) \cdot \T(\u_{m} - \u) \, d\Omega \label{eq:f1}
\eea
is minimized, and:
\beq
\bep = \nabla^s \u \label{eq:f2}
\eeq
and:
\beq
\nabla \cdot \A = 0 \qquad \mbox{a.e.} \qquad \mbox{in $\Omega$} \label{eq:f3}
\eeq
where:
\beq
\A(\bep) = 2\tr(\bep)\I + 2\bep. \label{eq:f13}
\eeq
and:
\beq
\nabla^s \u = \frac{1}{2}\left(\nabla\u + (\nabla\u)^{\mbox{T}}\right) \label{eq:f4}
\eeq
\end{problem}

With diagonal $\T$, equation (\ref{eq:f1}) may be written explicitly as:
\beq
\pi_{o}[\u] = \frac{1}{2} \int_{\Omega} [\mbox{T}_{xx}(u_{(m)x}-u_{x})^2 + \mbox{T}_{yy}(u_{(m)y} - u_{y})^2] \, d{\Omega}. \label{eq:f5}
\eeq
In equation (\ref{eq:f5}), the weights (T$_{xx}$, T$_{yy}$) allow more importance to be placed on the accurate component of the displacement field while calculating the predicted displacement field. In equation (\ref{eq:f3}), a.e. means \emph{almost everywhere}, and implies that the constraint is enforced everywhere in the domain except on a set of measure zero (i.e. along curves and at isolated points).

To motivate constraint (\ref{eq:f3}), we consider a thin sheet in the $(x, y)$ plane made of solid material that is linearly elastic, isotropic, incompressible, and simply connected. We further assume that the deformation of the material is small, quasistatic,  and approximately plane stress in character. For these modeling assumptions, linear strain $\bep$ and Cauchy stress $\bsigma$ are defined and related as:
\bea
\epsilon_{ij} &=& \frac{1}{2} (\partial_{i} u_{j} + \partial_{j} u_{i}) \qquad (i,j = 1,2) \label{eq:f6}  \\
\bsigma &=& -p\I + 2\mu\bep \label{eq:f7}
\eea
where $\partial_{i}=\frac{\partial}{\partial{x_{i}}}$. For a plane stress deformation:
\beq
\sigma_{zz} = -p + 2\mu\epsilon_{zz} = 0. \label{eq:f8}
\eeq

For an incompressible material:
\beq
\epsilon_{xx} + \epsilon_{yy} + \epsilon_{zz} = 0. \label{eq:f9}
\eeq

Substituting (\ref{eq:f9}) into (\ref{eq:f8}) gives:
\beq
p = -2\mu(\epsilon_{xx} + \epsilon_{yy}). \label{eq:f10}
\eeq

Using (\ref{eq:f10}) in (\ref{eq:f7}) yields:
\bea
\bsigma &=& 2\mu(\epsilon_{xx} + \epsilon_{yy})\I + 2\mu\bep \label{eq:f11} \\
        &=& \mu\A. \label{eq:f12} 
\eea

The equilibrium equation, with no body force, is $\nabla \cdot \bsigma = 0$, or using (\ref{eq:f12}):
\beq
\nabla \cdot (\mu\A) = 0 \qquad \mbox{in $\Omega$}. \label{eq:f14}
\eeq

If $\mu(x,y)$ is piecewise constant, then $\nabla \cdot \A = 0$ a.e. in $\Omega$. 

\subsection{Uniqueness of solution}
\label{ssec:analy}
Given an axial displacement measurement $u_y$, the conditions necessary to guarantee the uniqueness of the lateral displacements computed using the equilibrium equation constraints will be derived in this section. This derivation begins with the assumption that the material has a piecewise homogeneous shear modulus distribution. This means that the material's domain, $\Omega$, can be divided into n different subdomains, $\Omega^{(n)}$, and the shear modulus in each subdomain is constant. The shear modulus is not constant in the interface between the subdomains, however. An illustration of this type of domain, for a special case when $n=2$, is shown in figure (\ref{fig:domain1}). 
The conditions necessary for uniqueness of the lateral displacements in each subdomain will be derived next.


\subsubsection{Solution in a single subdomain}
\begin{lemma}\label{l1}
Given $u_y(x,y)$ in subdomain $\Omega^{(n)}$, in which $\mu \equiv$ constant.  
Then the general solution of equation (\ref{eq:f2}) - (\ref{eq:f4}) 
for $u_x(x,y)$, with $(x,y) \in \Omega^{(n)}$ is: 
\beq
u_x(x,y) = u_x^p(x,y) +  c_{_1}(x^2-4y^2) + c_{_2}x + c_{_3}y + c_{_4} \label{eq:a11p}
\eeq
\end{lemma}

\begin{proof}

Since $\mu = constant$ in $\Omega^{(n)}$,  we may expand equation (\ref{eq:f3})  in terms of displacements via equations (\ref{eq:f13}) and (\ref{eq:f2}), which results in the following equations:
\bea
-4u_{x,xx} - u_{x,yy} &=& 3u_{y,yx} \label{eq:a3} \\
-3u_{x,xy} &=& u_{y,xx} + 4u_{y,yy}. \label{eq:a4}
\eea

All the terms on the right hand side of both these equations are supposed to be known, since by assumption here, $u_y$ is known precisely. The terms on the left hand side involving the lateral displacements, $u_x$, are unknowns. We now show that these equations lead to a solution for $u_x$ that has at most four unknown constants.
We begin by splitting the total solution into homogeneous and particular parts.  
To that end, we let:
\beq
u_x = u_x^p + u_x^h \label{eq:a12}
\eeq
Here $u_x^p$ is any particular solution of (\ref{eq:a3}) and (\ref{eq:a4}), which, by assumption, has no free constants or undetermined functions.   If they appear in the process of solving (\ref{eq:a3}) and (\ref{eq:a4}), then they shall be set to arbitrary values (e.g.\ zero) without loss of 
generality.   Any part of the general solution (\ref{eq:a12}) that is undetermined by $u_y$ is thus contained in 
$u_x^h$, the homogeneous solution of (\ref{eq:a3}) and (\ref{eq:a4}).  
The homogeneous solution satisfies:
\bea
4u_{x,xx}^h + u_{x,yy}^h &=& 0 \label{eq:a13} \\
3u_{x,xy}^h &=& 0. \label{eq:a14}
\eea

Integrating (\ref{eq:a14}) twice yields:
\beq
u^h_x = f(x) + g(y). \label{eq:a6}
\eeq
Substituting (\ref{eq:a6}) into (\ref{eq:a13}) gives:
\beq
4f''(x) + g''(y) = 0. \label{eq:a7} 
\eeq
Equation (\ref{eq:a7}) implies:
\beq
4f''(x) = -g''(y) = 8c_{_1}. \label{eq:a8}
\eeq
Integrating (\ref{eq:a8}) for $f(x)$, and $g(y)$ yields:
\bea
f(x) &=& c_{_1}x^2 + c_{_2}x + c_{_4}' \label{eq:a9} \\
g(y) &=& -4c_{_1}y^2 + c_{_3}y + c_{_5}'. \label{eq:a10}
\eea
Combining (\ref{eq:a9}), (\ref{eq:a10}), in (\ref{eq:a6}) gives:
\beq
u_x^h = c_{_1}(x^2-4y^2) + c_{_2}x + c_{_3}y + c_{_4}, \label{eq:a11}
\eeq
where $ c_{_4} = c_{_4}' + c_{_5}'$. 

Therefore, the total solution is of the form (\ref{eq:a11p}).  
\end{proof}

Equation (\ref{eq:a11p}) shows that within any homogeneous subregion of the body, given the axial displacement $u_y(x,y)$, 
then the 
lateral displacement $u_x(x,y)$ has at most 4 undetermined constants.

\subsubsection{Solution in connected subdomains covering entire domain of interest}
Interestingly, the form of the homogeneous solution is 
completely independent of the given $u_y(x,y)$.   Thus we may conclude almost 
immediately:
\begin{lemma}\label{l2}
Given $u_y(x,y)$ in $\Omega= \cup \Omega^{(n)}$, in which $\mu(x,y)$ is piecewise constant, such that 
\beq
\mu(x,y) = 
\begin{cases} \mu^{(n)} &\mbox{if } (x,y) \in \Omega^{(n)} 
 \end{cases}.
\eeq
Then the general {\em global} solution 
of equations (\ref{eq:f2}) - (\ref{eq:f4}) 
for $u_x(x,y)$ with $(x,y) \in \Omega = \cup \Omega^{(n)}$
is:
\beq
u_x(x,y) = u_x^{p}(x,y) 
   +  c_{_1}(x^2-4y^2) + c_{_2}x + c_{_3}y + c_{_4} 
\eeq
Here, $u_x^{p}(x,y)$ is determined entirely by $u_y(x,y)$, while $c_{_1}, \ldots c_{_4}$ are undetermined constants.  
\end{lemma}

\begin{proof}
We consider several contiguous subregions, each of which is homogeneous.  From lemma \ref{l1}, the solution within subregion $n$ is given by:
\beq
u_x^{(n)}(x,y) = u_x^{(n)\ p}(x,y) +  c_{_1}^{(n)}(x^2-4y^2) + c_{_2}^{(n)}x + c_{_3}^{(n)}y + c_{_4}^{(n)}. \label{eq:a12p}
\eeq

If we consider the displacement $u_x(x,y)$ to be continuous at a non-special interface\footnote{A ``special interface''
would be one for which there exists $[c_1, c_2, c_3, c_4] = [c_1^*, c_2^*, c_3^*, c_4^*] \neq [0,0,0,0]$ s.t. 
$c_{_1}^{*}(x^2-4y^2) + c_{_2}^{*}x + c_{_3}^{*}y + c_{_4}^{*} = 0$ for all $(x,y)$ on the interface.}
between subregion $n$ and adjacent subregion $m$, and we assume that $u_x^{(n)p} = u_x^{(m)p}$ is continuous on the interface,
then we conclude: 
\beq
c_{_j}^{(n)} = 
c_{_j}^{(m)}.
\eeq
To see this, note that $u_x^{(n)} = u_x^{(m)}$, and $u_x^{(n)p} = u_x^{(m)p}$ leads to the equation:
\beq
\begin{bmatrix}
x^2-4y^2 & x & y & 1
\end{bmatrix}
\begin{Bmatrix} c_{_1}^{*} \\ c_{_2}^{*} \\ c_{_3}^{*} \\ c_{_4}^{*} \end{Bmatrix} = 
\begin{Bmatrix} 0 \end{Bmatrix} \label{eq:a12.1p}
\eeq
with $c_{_j}^* = c_{_j}^{(n)} - c_{_j}^{(m)}$. We note that (\ref{eq:a12.1p}) holds at every point on the interface. Provided that  the functions $x^2-4y^2$, $x$, $y$, $1$ are linearly independent on the interfaces, then the only solution of (\ref{eq:a12.1p}) is $c_{_j}^* = 0$. Interfaces on which these functions are linearly dependent are special cases, and we call them ``special interfaces."

Since the $c_j$'s take the same value in every subregion, we may consider the homogeneous functions to be defined globally.  
Thus, the global solution is:
\beq
u_x(x,y) = u_x^{p}(x,y) 
   +  c_{_1}(x^2-4y^2) + c_{_2}x + c_{_3}y + c_{_4} \label{eq:a13p}
\eeq
\end{proof}
Lemma \ref{l2} and equation (\ref{eq:a13p}) implies that by knowing the axial displacement throughout the domain in a piecewise homogeneous material, the lateral displacement may be determined up to four coefficients.  Hence, if sufficient information exists in the measurements to determine these four coefficients accurately, then the lateral displacement field may be determined accurately throughout the domain.  

In the next section, we derive a variational formulation that permits this estimation in a natural way.

\subsection{SPREME variational formulation}
\label{ssec:VarFomu}
In practice, we shall enforce constraint (\ref{eq:f3}) approximately as a penalty, in the form of a regularization term. To determine the appropriate form of that term, we refer to Figure (\ref{fig:domain1}) and consider a \emph{continuous} $\mu$ of the form:
\begin{figure}[!ht]\centering
      \includegraphics*{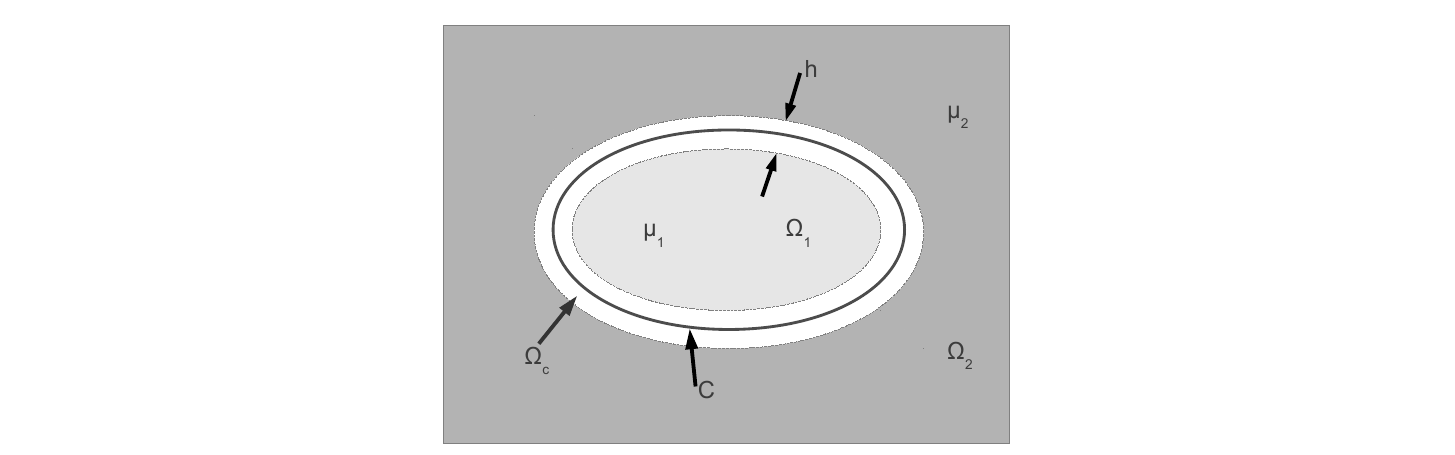}
  \caption{A piecewise continuous distribution of elastic modulus represented as the limit of a continuous distribution. 
The function $\mu(x,y)$ has the constant value $\mu_1$ for $(x,y) \in \Omega_1$, 
has the constant value $\mu_2$ for $(x,y) \in \Omega_2$, 
and is continuous and monotone in a narrow intermediate region, $\Omega_c$.}
  \label{fig:domain1}
\end{figure}
\beq
\mu = 
\begin{cases}
  \mu_1       & x\in \Omega_1 \\
  \mu_2       & x\in \Omega_2 \\
  \text{continuous monotone}      & x\in \Omega_c
\end{cases}
\eeq
where $\Omega_c$ is the set of all points within $\frac{h}{2}$ from curve C. We assume there is a constant $C_1$ so that $|\nabla \mu| \le \frac {C_1}{h}$ for all $x \in \Omega_c.$  

Suppose $\nabla \cdot (\mu\A) = 0$ everywhere in $\Omega$, then:
\beq
\nabla \cdot (\mu\A) = 
\begin{cases}
  \mu_1 \nabla \cdot \A = 0      & x\in \Omega_1 \\
  \mu_2 \nabla \cdot \A = 0      & x\in \Omega_2 \\
  \mu \nabla \cdot \A + \A \nabla \mu = 0     & x\in \Omega_c
\end{cases}. \label{eq:f25}
\eeq

In order to form the penalty term, we assume (\ref{eq:f25}) is satisfied and consider:
\bea
\int_{\Omega} |\nabla \cdot \A|^{\alpha} \, d\Omega = \int_{\Omega_c} \left|\A  \frac{\nabla \mu}{\mu}\right|^\alpha d\Omega_c &\le&
\int_{\Omega_c} \frac{|\A|^\alpha |\nabla\mu|^\alpha}{(\mu_{min})^\alpha} d\Omega_c
\le \int_{\Omega_c} \left(\frac{C_2}{h}\right)^{\alpha} d\Omega_c
 \label{eq:f15}\\
\sim h L_{(c)}\left(\frac{C_2}{h}\right)^{\alpha} &=& L_{(c)}C_2^{\alpha}h^{1-\alpha} \label{eq:f16} 
\eea
where $\mu_{min} =$ min$\{\mu_1,\mu_2\}$ is the minimum value of $\mu$ in $\Omega_c$, $C_2=(C_1 |\A|)/\mu_{min}$, $|\A|$ is the largest eigenvalue of $\A$, and $L_{(c)}$ is the perimeter of curve C; see Figure (\ref{fig:domain1}).

We see in (\ref{eq:f16}) that for piecewise constant $\mu$ (i.e. when h $\rightarrow 0$), 
\beq
\int_{\Omega} |\nabla \cdot \A|^\alpha \, d\Omega = 0 \qquad \mbox{for} \qquad 0 < \alpha < 1. \label{eq:mnc16}
\eeq
Equation (\ref{eq:mnc16}) shows that we may enforce the constraint $\nabla \cdot \A =0\quad a.e.$ by choosing to enforce the integral condition (\ref{eq:mnc16}) with $0< \alpha < 1$.

The derivative of the absolute value function is undefined at the origin, so in practice we introduce a small positive constant, $\delta$, and approximate (\ref{eq:mnc16}) as:
\beq
\int_{\Omega} |\nabla \cdot \A(\bep)|^{\alpha} \, d\Omega \approx \int_{\Omega}
\frac{\left(\nabla \cdot \A(\bep^{{k}})\right)^{2}}{\Big[\left(\nabla \cdot \A(\bep^{^{(k-1)}})\right)^{2} + \delta\Big]^n} \, d{\Omega}  \label{eq:f17}
\eeq
where:
\beq
n = 1 - \frac{\alpha}{2} 
\eeq
Note that for $0 < \alpha < 1 \Longleftrightarrow 1 > n > \frac{1}{2}$.

In Equation (\ref{eq:f17}), k denotes an iteration counter. As k approaches infinity, we expect $\bep^{^{(k-1)}}$ to approach $\bep^{{k}}$. In this limit, for $\delta = 0$, the approximation in (\ref{eq:f17}) is then exact. The $\delta$ in the denominator is a small numerical constant that prevents the denominator from being zero when $\nabla \cdot \A = 0.$

Using (\ref{eq:f17}) instead of (\ref{eq:f3}), we re-formulate the problem as:

\begin{problem}
Given $\u_{m}(x,y)$, $\x\in\Omega$ and $\bep^{{(k-1)}}(x,y)$, find $\u^k$, and $\bep^{{k}}$ that minimizes:
\beq
\pi[\u, \bep] = \pi_{_{O}} + \pi_{_{C}} + \pi_{_R} \label{eq:f18}
\eeq
where:
\beq
\pi_{_{O}} = \frac{1}{2} \int_{\Omega} (\u_{m} - \u^k) \cdot \T(\u_{m} - \u^k) \, d\Omega
\eeq
and:
\beq
\pi_{_{C}} = \frac{\beta}{2} \int_{\Omega} (\bep^k - \nabla^s \u^k)^2 \, d{\Omega} \label{eq:f19}
\eeq
and:
\beq
\pi_{_{R}} = \frac{1}{2}\int_{\Omega} \alpha(\x) \left(\nabla \cdot \A(\bep^{{k}})\right)^{2} \, d{\Omega} \label{eq:f20}
\eeq
and:
\beq
\alpha(\x) = \frac{\alpha_o}{\Big[\left(\nabla \cdot \A(\bep^{^{(k-1)}})\right)^{2} + \delta\Big]^n}.
\eeq
\end{problem}
This yields a quadratic minimization problem at each iteration. We initialize iterations using $\bep^{(0)}=0$.  
Equation (\ref{eq:f18}) contains a displacement data matching term ($\pi_{_{O}}$), a compatibility term ($\pi_{_{C}}$), and a regularization term ($\pi_{_{R}}$).  $\pi_{_{O}}$ minimizes the mismatch between the measured displacement field and the filtered displacement field,
 $\pi_{_{C}}$ 
constrains the calculated strains to satisfy the compatibility equation. $\beta$ is a constant that is used to regulate how strongly the compatibility condition is enforced. $\pi_{_{R}}$ is used to constrain the strains to satisfy the equilibrium equations almost everywhere in the domain. $\alpha({\x})$ maybe interpreted as a spatially varying regularization constant that regulates how strongly the equilibrium equation is enforced locally. 

The way the regularization, $\pi_{_{R}}$, works can be best explained with figure (\ref{fig:domain1}). In this figure, where $\mu$ is constant (i.e. in the background ($\mu_2$) and in the inclusion ($\mu_1$)), $\nabla\cdot\A(\bep^{^{\left(k-1\right)}})$ will be small meaning that $\alpha(\x)$ will be large, so the equilibrium equation will be strongly enforced in these regions. Therefore, $u_x$ can be confidently determined there up to the homogeneous solution. Where $\mu$ is not constant (i.e. in $\Omega_c$, a thin region along the interface between the inclusion and the background), $\nabla \cdot \A(\bep^{^{\left(k-1\right)}})$ will be large, meaning that $\alpha(\x)$ will be small, so the equilibrium equation will not be strongly enforced. Thus the strains are permitted to change rapidly there. Continuity of the lateral displacements is enforced there through the $\pi_{_{C}}$ term. This ``spatially adaptive'' regularization thus allows us to enforce the equilibrium equation everywhere in the domain, except near material interfaces.


We note that if the goal is to evaluate the strain tensor only, then it is tempting to reformulate the problem without the displacement variable, $\u$. Such a formulation is presented in appendix (\ref{app:1stFormu}). There, it is shown that the strains obtained from the formulation are not compatible, and hence not physical, for some simulation experiments. 

Equation (\ref{eq:f18}) is the \spreme\ functional. Our implementation minimizes $\pi[\u,\bep]$ by the finite element method (FEM) discretization. The weak form is derived first by computing the following directional derivatives:

\bea
D_{_{\u}} \pi : \v &=& \frac{d}{dc} \Bigg |_{c = 0} \pi [\u + c\v, \bep] = 0 \\ 
D_{_{\bep}} \pi : \w &=& \frac{d}{dc} \Bigg |_{c = 0} \pi [\u, \bep + c\w] = 0 \label{eq:s0} 
\eea
where  $\v$ and $\w$ are admissible variations of $\u$ and $\bep$, respectively, and $\bep = \bep^{k}$. Taking the directional derivatives yields the following system of equations:

\beq
\int_{\Omega} \T(\u - \u_{m})\v \, d\Omega - \beta \int_{\Omega} \left[\bep - \nabla^s \u \right](\nabla \v) \, d\Omega = 0 \qquad \mbox{for all } \v \label{eq:s1}
\eeq


\beq
\beta \int_{\Omega} \left[\bep - \nabla^s \u \right]\w \ d\Omega + \int_{\Omega} \alpha(\x) \{\nabla \cdot \A(\w)\} \cdot \{\nabla \cdot \A(\bep)\} \, d{\Omega} = 0 \qquad \mbox{for all } \w \label{eq:s2}
\eeq

Equations (\ref{eq:s1}) and (\ref{eq:s2}) were discretized with bilinear shape functions to yield a system of linear equations. A new finite element within an in-house FEM code was created to solve the system of linear equations following standard FE methods e.g. \cite{book:tjrhughes}.

\section{Verification of computational implementation}
\label{sec:veri}
The purpose of verification is to confirm that, given data that satisfies the modeling assumptions, a computational implementation produces the correct solution that is consistent with the assumed model.  In this section, therefore, the results of two different simulation experiments conducted to test the implementation of the \spreme\ formulation will be described. For the two experiments performed, reference analytical displacement and strain fields are available, and they were used to evaluate the quality of reconstructed displacements and strains from the \spreme\ formulation. 

\subsection{Reference analytical solution}
An exact analytical solution of the elasticity equations is used to generate reference displacement data. The solution corresponds to that of an infinite elastic sheet with a perfectly bonded circular inclusion.  The sheet is subject to uniform stress at infinity. 
The problem was nondimensionalized using the region of interest size as the reference length, and the background shear modulus as the reference stress.
 The shear modulus of the inclusion and background were $\mu_{_I}=3$, and $\mu_{_B}=1$, respectively. The Poisson's ratio was taken to be $\nu = 0.5$, consistent with the incompressibility assumption. The radius of the inclusion was $a=0.25$. The solution to the elasticity equation was obtained using the methods described in \cite{honein:thesis}.

\subsection{Uniaxial tension test}
\label{ssec:uniTen}
\begin{figure}[!ht]\centering
      \includegraphics*{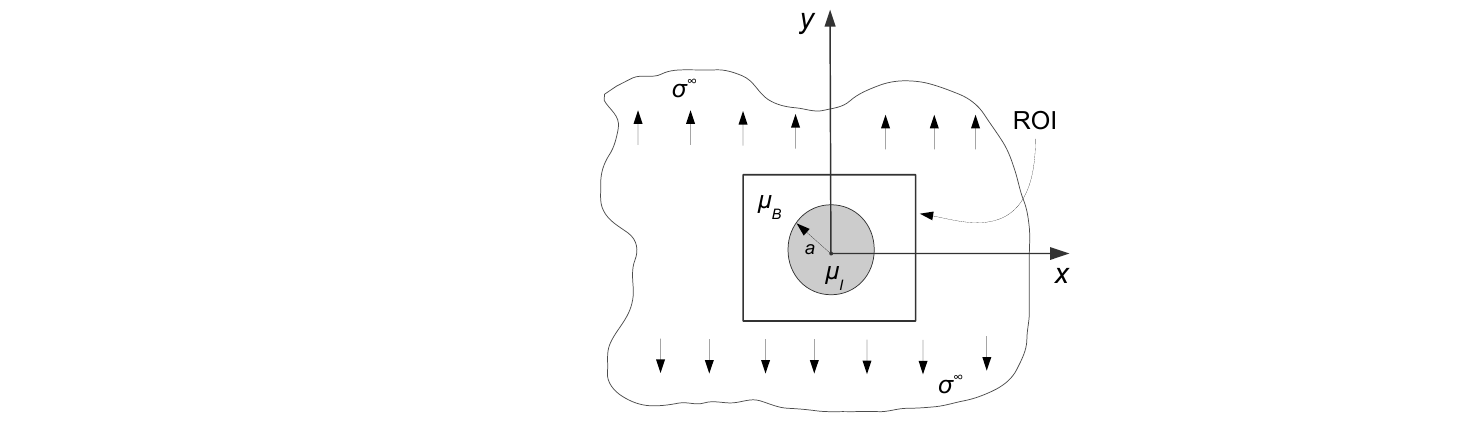}
  \caption{Uniaxial tension test: 
The exact analytical solution for far-field uniaxial tension applied to an infinite sheet with a bonded 
circular inclusion is used for verification of the computational implementation.  
In the figure, ROI indicates the Region Of Interest over which displacement fields are assumed to have been ``measured".} 
  \label{fig:numEx}
\end{figure}

\begin{figure}[!t]
\centering
      \includegraphics*{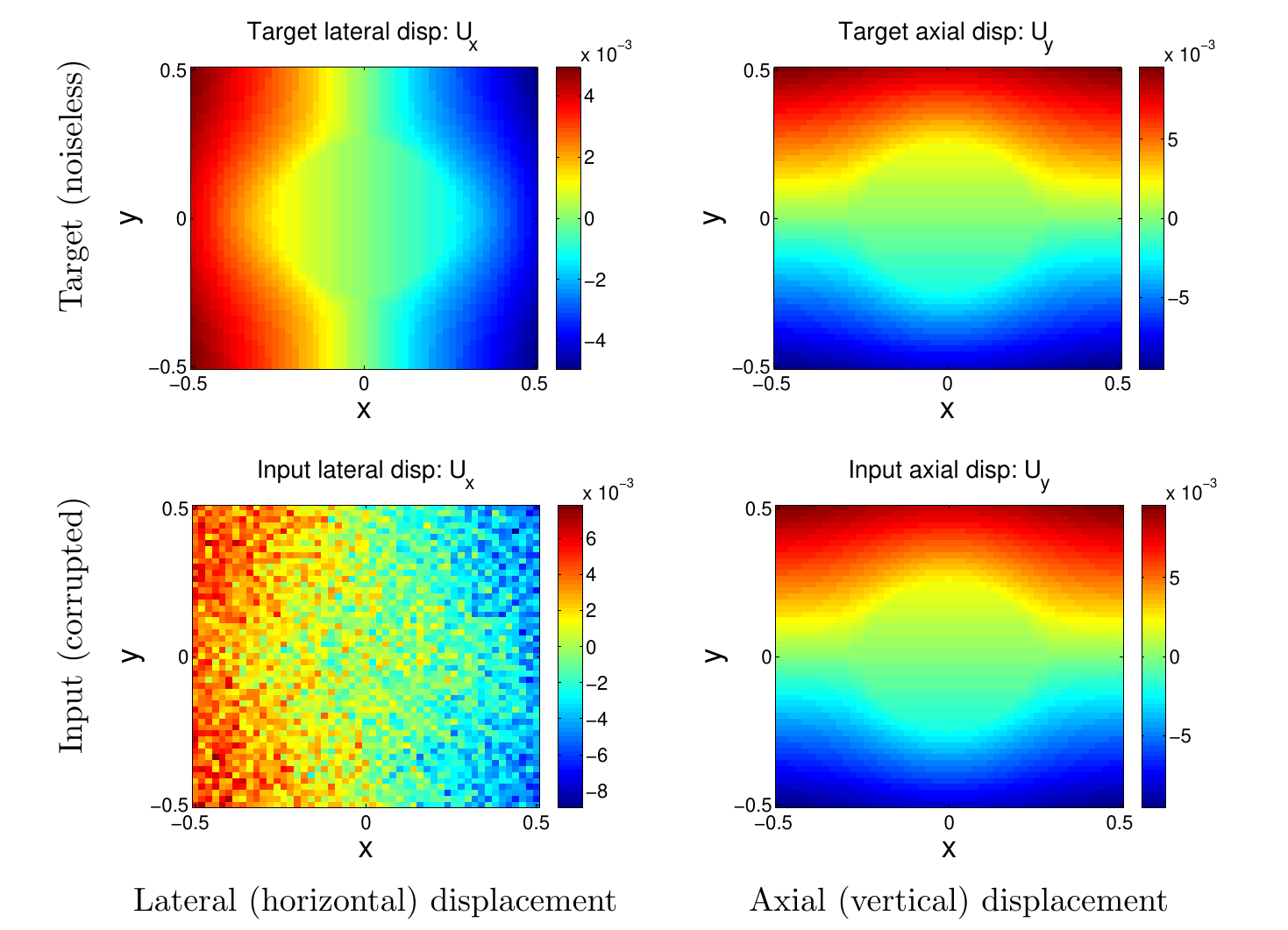}
  \caption{Uniaxial tension test.  Upper row: Exact noiseless displacement field for uniaxial tension experiment.  
These represent the target displacement fields.  Lower row:  
Noise-corrupted displacement fields used for verification. 
The noise level in the lateral displacements is similar to that seen in physically 
measured displacement fields shown in figure \ref{fig:latDisp}.}
\label{fig:targDisp}
\end{figure}
%

\begin{figure}[!t]
\centering
      \includegraphics*{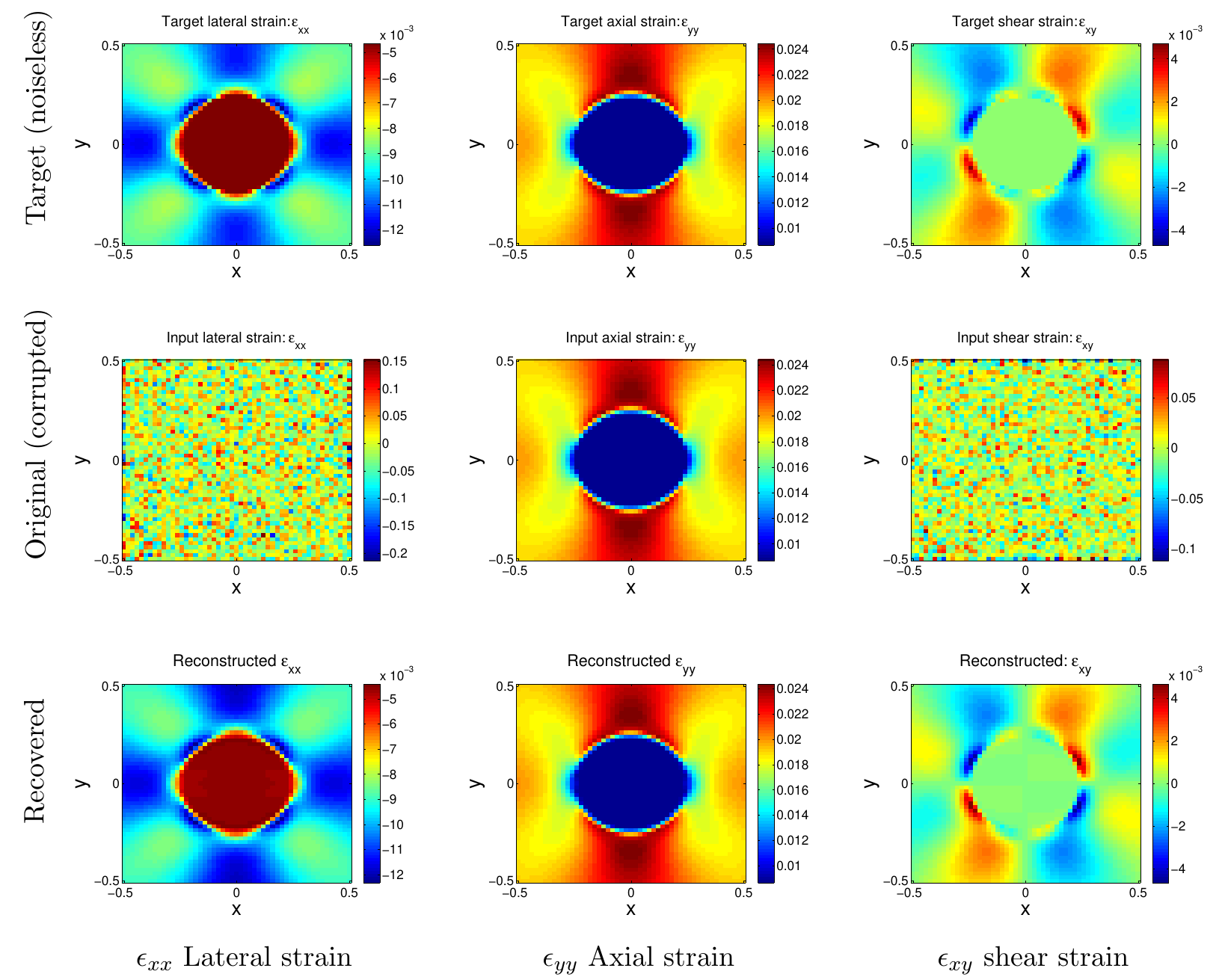}
  \caption{Uniaxial tension test results.  
Top row: Target noiseless strain fields from uniaxial tension test.  
Middle row:  Strain fields calculated from noise-corrupted input displacement fields. 
Bottom row:  Strain fields recovered from \spreme.
}
\label{fig:test1strain}
\end{figure}


For the first experiment the sheet was subject to uniaxial tension in the $y$ direction as depicted in figure (\ref{fig:numEx}). 
The stress at infinity was taken to be $\sigma_{yy} = \sigma^{\infty} = 6 \times 10^{-2}$.  
The displacement field corresponding to this deformation was evaluated in a subregion around the inclusion. 
This subregion was discretized into $50 \times 50$ square elements in the $x$ and $y$ directions, 
and the axial and lateral displacements were evaluated at each of the $2601$ nodes.  
This reference displacement field is shown in figure (\ref{fig:targDisp}). 
This target displacement field was differentiated by finite differences 
to calculate the target strains shown in figure~(\ref{fig:test1strain}).  
Since the analytical expressions for the displacements are available, the target strains can be differentiated exactly 
rather than using finite differences.  
There is no analytical expression for the noise corrupted displacement fields, however, and therefore we have chosen 
to treat both noisy and reference fields in the identical manner.  
These target fields will be used as a reference to benchmark the accuracy of the fields reconstructed with \spreme. 

Our problem assumptions include that we are given accurate axial displacements and noisy lateral displacements.  Therefore, Gaussian noise was added to the lateral displacement ($u_x$) in order to simulate ultrasound displacement measurements, and the axial displacement ($u_y$) was left unchanged. The corrupted displacement field is shown in figure (\ref{fig:targDisp}).

The corrupted displacement field was differentiated by finite differences to obtain the strains displayed in figure (\ref{fig:test1strain}). The $\ep_{yy}$ strain component was unchanged. The features in the other strain components, however, were completely obscured by noise. 

The corrupted displacement field was processed with \spreme\ using the following parameter values: 
$\alpha_o = 10^{-5}$, $\beta = 1$, $\delta = 10^{-8}$, $n=0.5$, $T_{xx} = 10^{-9}$, $T_{yy} = 10^{4}$. 
These parameters were chosen by trial and error from extensive computational testing.  
This same trial and error method was used to choose the best 
parameters for all the other computational experiments performed.  
Six iterations were performed to allow the strains to converge.  
The strains from \spreme\ are displayed in the bottom row of figure (\ref{fig:test1strain}). 
The reconstructed strains are very smooth because of the regularization on the momentum equation. 
This regularization term is very similar to a Total Variation (TV) regularization which tends to 
penalize reconstructions with high oscillations, but allows reconstructions to have 
jumps that need to be present to fit the data \cite{vogel:IP}.  

Accuracy of reconstructed strains and displacements may be seen 
qualitatively in line plots passing through the center of the inclusion region, figure (\ref{fig:line_Uaxl}).  
These plots show excellent agreement between the target strains and 
the reconstructed strains (referred to as \spreme\ Strain in the legend) in every component of strain.  
We recall that the \spreme\ formulation has independent variables for both 
displacements and strains, and therefore we also use this figure to 
evaluate consistency of the \spreme\ displacements and strains.  
To see this, we use a  finite difference approximation of the 
derivative of the reconstructed displacements to obtain 
another estimate of computed strains; these are labeled 
\spreme\ Disp in the legend, and are also in excellent 
agreement with the other quantities.  
\begin{figure}[!ht]
\centering
        \includegraphics*{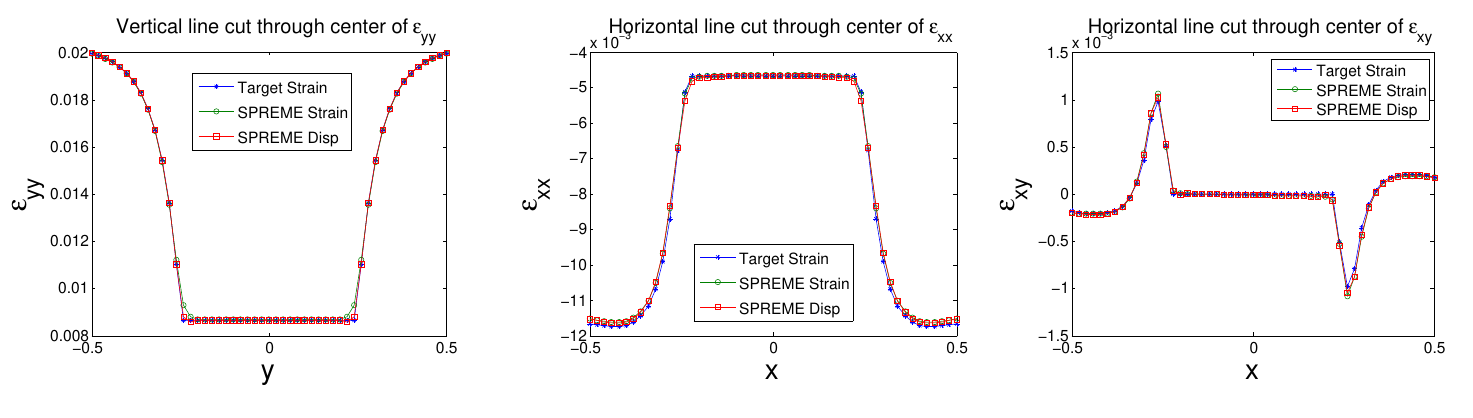}
  \caption{Uniaxial tension: Line plots through the axial, lateral, and shear strain fields. 
The \spreme\ method produces independent estimates of displacement and strain.  Thus 
each graph shows three curves.  One is the target strain.  The second is the strain output 
from \spreme\, labeled ``\spreme\ Strain." The third is the derivative of 
the output \spreme\ displacement, labeled 
``\spreme\ Disp".  Comparing the latter two provides a test for 
consistency between these fields.  
}
  \label{fig:line_Uaxl}
\end{figure}

Overall quantitative accuracy of the displacements (strains) computed with \spreme\ were evaluated for both individual components and for the total vector (tensor).  
The component error is defined as: 
\beq
\mbox{Strain error} = \frac{||\ep_{ij}^t - \ep_{ij}||}{||\ep_{ij}^t||} \approx \frac{\sqrt{\sum_{a=1}^{2601} (\ep_{ij(a)}^t - \ep_{ij(a)})^2}}{\sqrt{ \sum_{a=1}^{2601} (\ep_{ij(a)}^t)^2 }} \label{eq:eNorm_str}
\eeq

\beq
\mbox{Disp error} = \frac{||u_i^t - u_i||}{||u_i^t||} \label{eq:eNormD}
\eeq
Here $u_i^t$, and $\ep_{ij}^t$  are the target displacement and strain components, respectively, and the $i's$ and $j's$ go from 1-2 with 1 representing the $x$ direction, and 2 representing the $y$ direction.   The counter $a$ runs over all nodes in the mesh.

The norms calculated with these formulas were plotted as a function of iteration number in figure~(\ref{fig:norms}). For iteration zero, $\u$ was the displacement field corrupted by noise, and $\bep$ was initialized to zero strain.  The plots in the figure demonstrate that the errors reduce significantly with iteration number.  The exact values of the errors in the field variables after the sixth iteration are displayed in table~(\ref{tab:errVals}). From the table, we see that the error in $u_x$ (resp/.\ $\epsilon_{xx}$) reduced from about 50\% (resp. $500\%$) to about 2\%.  
\begin{figure}[!ht]
\centering
          \includegraphics*{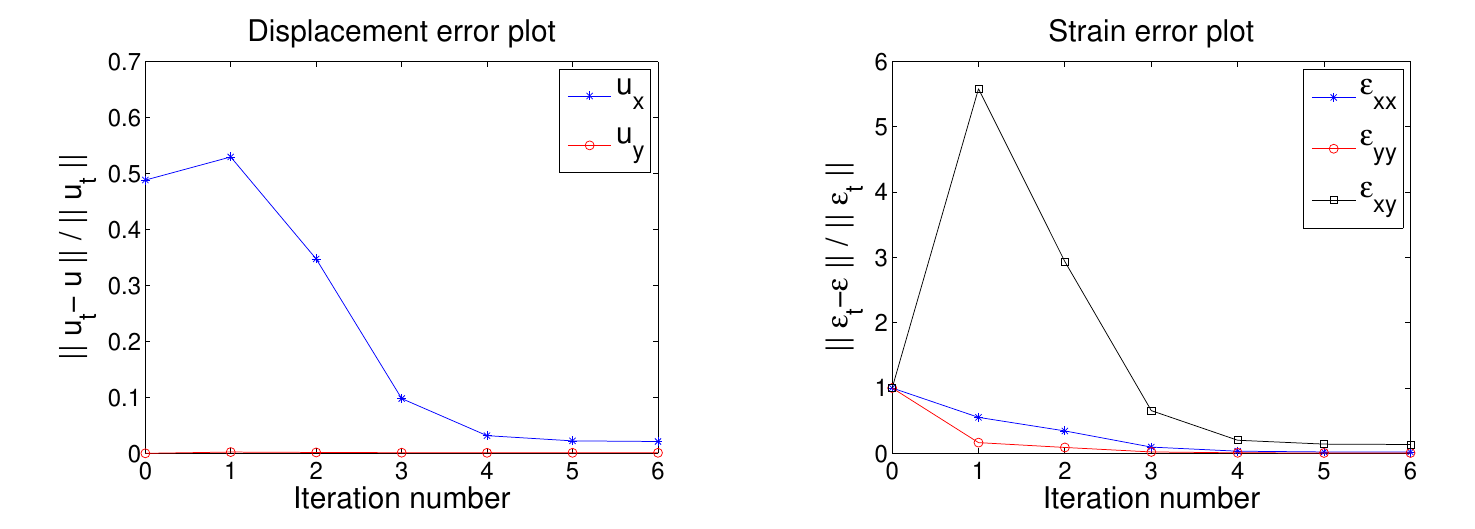}
  \caption{Uniaxial tension test: Error in each component of the field variables as a function of iteration number.}
  \label{fig:norms}
\end{figure}

\begin{figure}[!ht]
          \includegraphics*{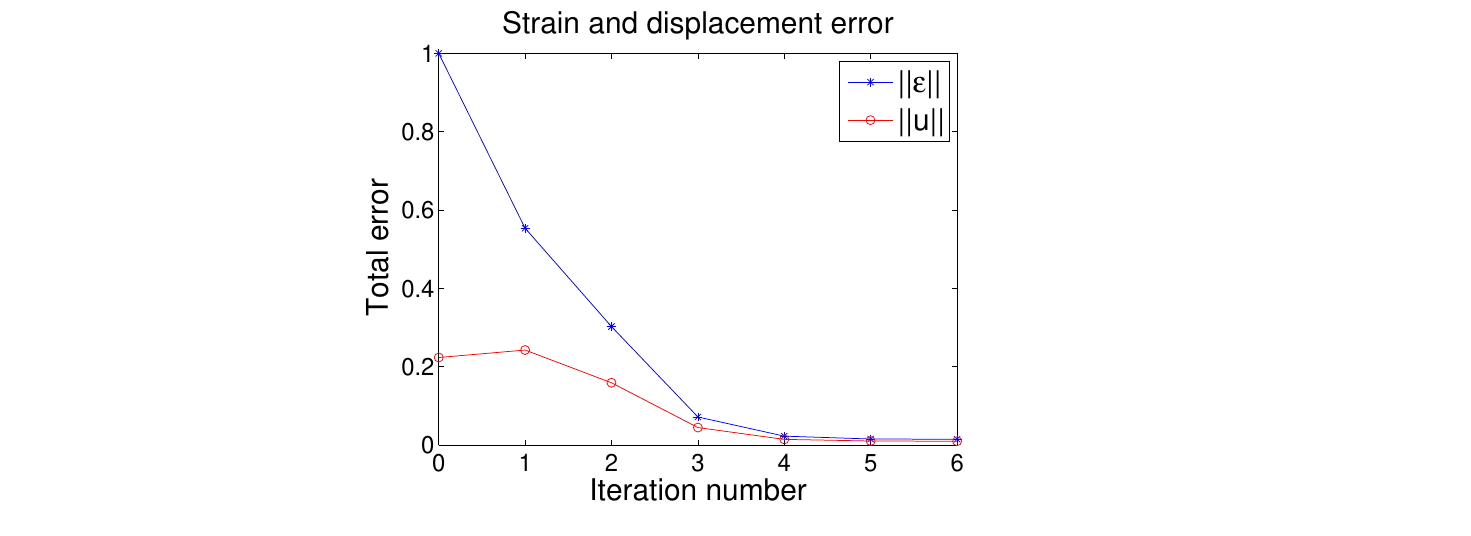}
  \caption{Uniaxial tension test: Total error in the field variables as a function of iteration number.}
  \label{fig:totNorm}
\end{figure}

The total error in strain and displacement were computed with the following formulas:
\bea
\mbox{Total strain error} &=& \sqrt{\frac{||\ep_{xx}^{err}||^2 + ||\ep_{yy}^{err}||^2 + 2||\ep_{xy}^{err}||^2}{||\ep_{xx}^t||^2 + ||\ep_{yy}^t||^2 +2||\ep_{xy}^t||^2}} \label{eq:eNorm_str_tot} 
\\
\mbox{Total disp error} &=& \sqrt{\frac{||u_x^{err}||^2 + ||u_y^{err}||^2}{||u_x^t||^2 + ||u_y^t||^2}} \label{eq:eNormD_tot}
\eea
Here $u_i^{err} = u_i-u_i^t$ and $\ep_{ij}^{err} = \ep_{ij} -\ep_{ij}^t$. A plot of these total errors as a function of iteration number is shown in figure (\ref{fig:totNorm}), which shows that the total errors decrease with iteration. 

\begin{table}[!ht]
\begin{center}
        \begin{tabular}{ | p{3.0cm} | p{4.0cm} | c | }
        \hline
        \centering Field variable & \centering Error in input field (\%) & Error in reconstructed field (\%)\\ \hline
        \centering $u_x$ & \centering 48.8 & $2.09$ \\ \hline
        \centering $u_y$ & \centering 0 & $0.072$ \\ \hline
        \centering $||\u||$ & \centering 22.6 & $0.759$ \\ \hline        
        \centering $\ep_{xx}$ & \centering 511 & $1.78$ \\ \hline
        \centering $\ep_{yy}$ & \centering 0 & $0.526$ \\ \hline
        \centering $\ep_{xy}$ & \centering 1883 & $13.5$ \\ \hline
        \centering $||\bep||$ & \centering 277 & $1.49$ \\ \hline        
        \end{tabular}
\end{center}
\caption{Uniaxial tension test: Errors in the recovered field variables after the sixth iteration, and initial errors in the input field variables.}
\label{tab:errVals}
\end{table}

\subsection{Biaxial tension test}
For the second verification test, we consider the same body, this time under biaxial tension; c.f.\  
figure (\ref{fig:numEx1}). The stress at infinity was $\sigma_{xx} = \sigma_{yy} = \sigma^{\infty} = 6 \times 10^{-2}$. 
\begin{figure}[!ht]\centering
          \includegraphics*{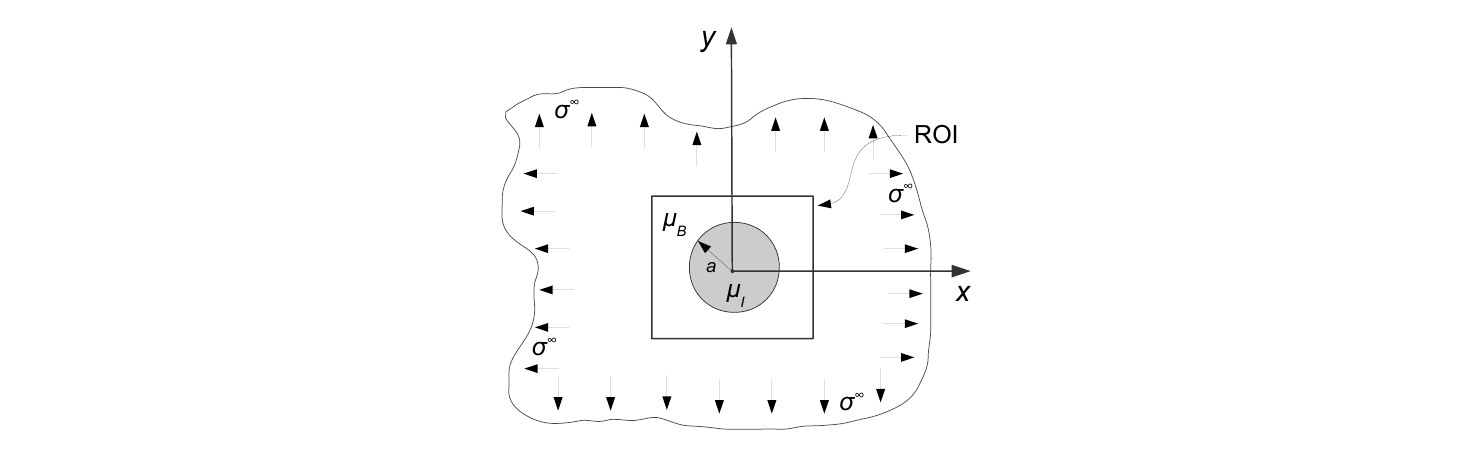}
  \caption{Biaxial tension test: 
A second verification experiment uses 
the exact analytical solution for far-field biaxial tension applied to an infinite sheet with a bonded 
circular inclusion. }  
  \label{fig:numEx1}
\end{figure}
Next, Gaussian noise was added to the lateral displacement field, and reconstructions were performed using the following parameters: $\alpha = 10^{-5}$, $\beta = 10$, $\delta = 10^{-8}$, $n=0.5$, $T_{xx} = 50$, $T_{yy} = 5 \times 10^{4}$. These are approximately the same parameters as the last experiment, except with $T_{xx}$ significantly increased to accommodate the applied horizontal stress.

\begin{figure}[!t]
\centering
          \includegraphics*{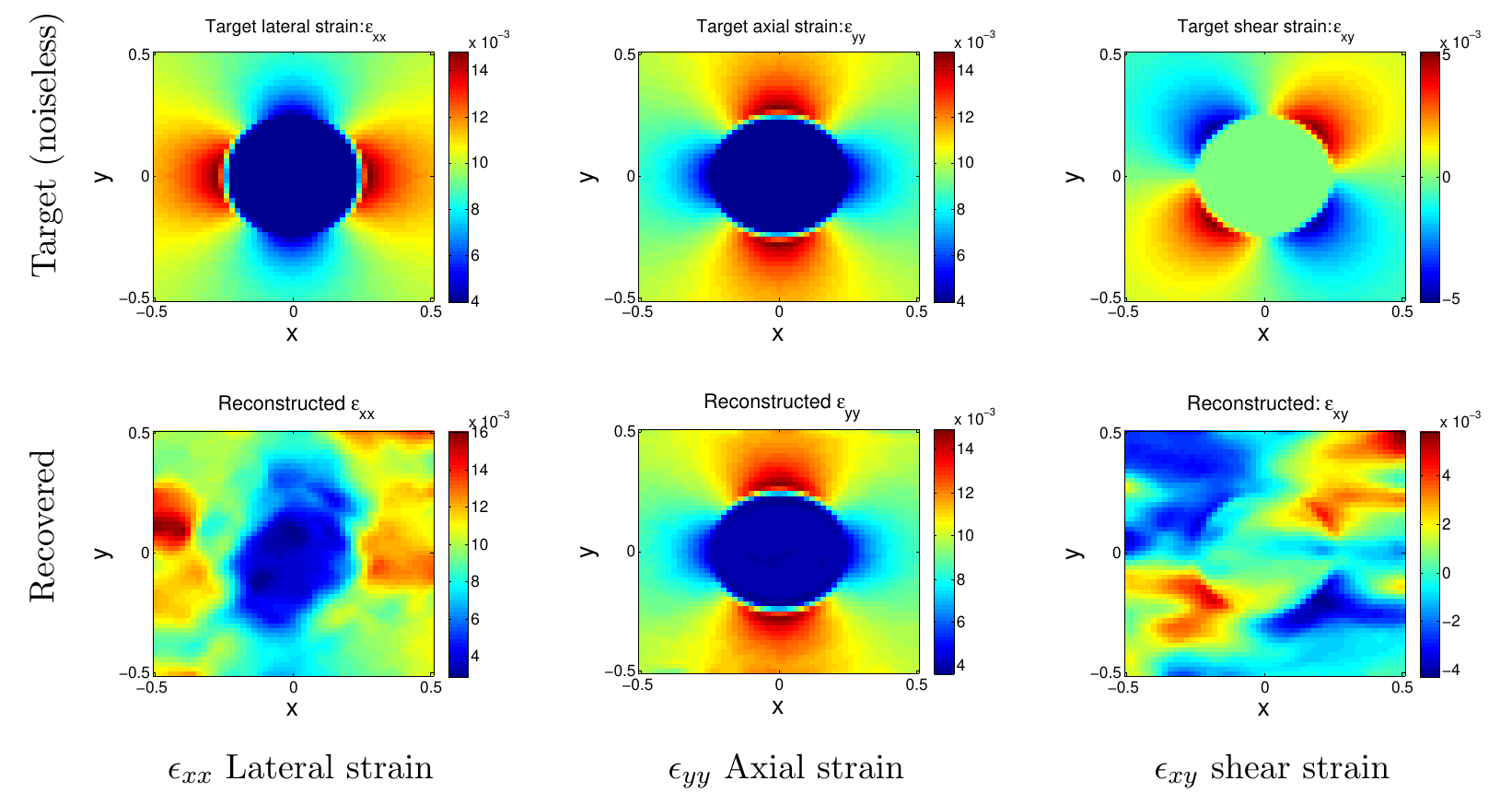}
  \caption{Biaxial tension test results. 
Upper row: Target noiseless strain fields from uniaxial tension test.  
Lower row: Strain fields recovered from \spreme.  
Relatively large weights on noisy lateral displacement increase the error in the 
recovered strains relative to those in Figure \ref{fig:test1strain}. 
}
\label{fig:test2strain}
\end{figure}



The target and reconstructed strains at the sixth iteration are shown in figure (\ref{fig:test2strain}). 
The features in the lateral and shear strains are recovered, but we see 
significantly more error in this test case than in figure (\ref{fig:test1strain}). 
That error is quantified 
in table (\ref{tab:errVals_biAx}). 
The table shows clearly that the errors in the lateral displacement, and lateral and shear strains reduced significantly 
from the errors in the input field.   This means that the method is able to work also for this type of deformation. 
The method did not work as well as it did for the case when the domain was under uniaxial tension, however.  
This is most likely due to the need to use relatively large weights on the noisy 
lateral displacements to capture the effect of significant lateral normal stress applied at infinity.  
Thus was the effect of the noise in that displacement component amplified.  Nevertheless, 
the error in $\ep_{xx}$ reduced from about $500\%$ to about $15\%$, a very significant reduction.  

\begin{table}[!ht]
\begin{center}
        \begin{tabular}{ | p{3.0cm} | p{4.0cm} | c | }
        \hline
        \centering Field variable & \centering Error in input field (\%) & Error in reconstructed field ($\%$) \\ \hline
        \centering $u_x$ & \centering 51.2 & $6.63$ \\ \hline
        \centering $u_y$ & \centering 0 & $0.103$ \\ \hline
        \centering $||\u||$ & \centering 36.6 & $3.83$ \\ \hline        
        \centering $\ep_{xx}$ & \centering 495 & $15.4$ \\ \hline
        \centering $\ep_{yy}$ & \centering 0 & $1.04$ \\ \hline
        \centering $\ep_{xy}$ & \centering 1370 & $59.3$ \\ \hline
        \centering $||\bep||$ & \centering 422 & $13.9$ \\ \hline        
        \end{tabular}
\end{center}
\caption{Biaxial tension test: Errors in the recovered field variables after the sixth iteration, and in the input field variables for the biaxial tension experiment with noise.}
\label{tab:errVals_biAx}
\end{table}

\section{Validation}
\label{sec:vald}
In the verification section, we demonstrated that \spreme\ correctly reconstructs plane stress displacement fields. Here in the validation section, we demonstrate that the \spreme\ implementation of the plane stress model applies to physical systems of interest. The displacement data used for the test was measured within a gelatin phantom under uniaxial compression as described in \cite{jDordLNLinear}.  This displacement data was analyzed previously by an iterative elastic modulus inversion algorithm \cite{jDordLNLinear}.  The results from that study, and a smoothed version of the measured displacements, will be used as a reference to compare to the results produced by \spreme. 

\subsection{Phantom tests}

The displacement data used for the test was measured within a gelatin phantom under uniaxial compression as described in \cite{jDordLNLinear}. 
 Here, for completeness of presentation, we briefly describe the data collection procedure. 
An agar and gelatin phantom was manufactured of to have the acoustic and mechanical properties of soft tissue using methods described in \cite{pavan2010nonlinear}.  The phantom was a 100mm cube containing four spherical inclusions of diameter $10mm$. The inclusion centers are coplanar in a horizontal plane separated by $30mm$ center to center distance.  Each of the inclusion has a different material Young's modulus \cite{jDordLNLinear}.

\subsubsection{Ultrasound imaging and displacement estimation}
The phantom was imaged with a Siemens SONOLINE Antares (Siemens Medical Solutions USA, Inc, Malvern, PA) clinical ultrasound system, with a linear ultrasound transducer array (Siemens VFX9-4) which was pulsed at 8.89MHz \cite{jDordLNLinear}. 
A 15cm $\times$ 15cm compression plate, much larger than the phantom surface, was attached to the ultrasound transducer to help generate approximately uniaxial deformation of the phantom. The phantom was first imaged before any deformation was applied, and it was imaged after every 0.5\% strain up until a total strain of 20\% with respect to the phantom's height \cite{jDordLNLinear}. 
During the imaging, Radio Frequency (RF) data, representing the spatial distribution of the backscattered pressure field,
was recorded. A modified block matching algorithm was used to estimate the displacement field from the RF data \cite{jiang2011fast}.   Because the models used in this paper depend on the small strain assumption, we use only the measured deformation fields corresponding to small strain, specifically, 
we choose $1.5\%$ overall strain.

\subsubsection{Reference results}
\label{sssec:refRes}
We compare our reconstructed displacement fields to two reference fields.  

The first reference displacement is simply the 
spatially smoothed version of the measured lateral displacement field. This smoothing is achieved by performing a local spatial average with corrections at the boundaries.  A smoothed displacement field reveals its main features, but, of course, obscures detail.

A second reference displacement field is computed as the result of a second inverse problem 
to reconstruct the linear elastic shear modulus, as described in \cite{goenezen2012linear}.
To that end, given a current guess of the modulus field, a forward elasticity problem is solved driven by assumed boundary conditions.  The modulus field is updated to minimize the difference between 
measured and predicted displacements.  Only axial displacements are used in the minimization.

%

\subsubsection{Phantom results and discussion}
\label{sssec:phantResDis}
\begin{figure}[!t]
\centering
          \includegraphics*{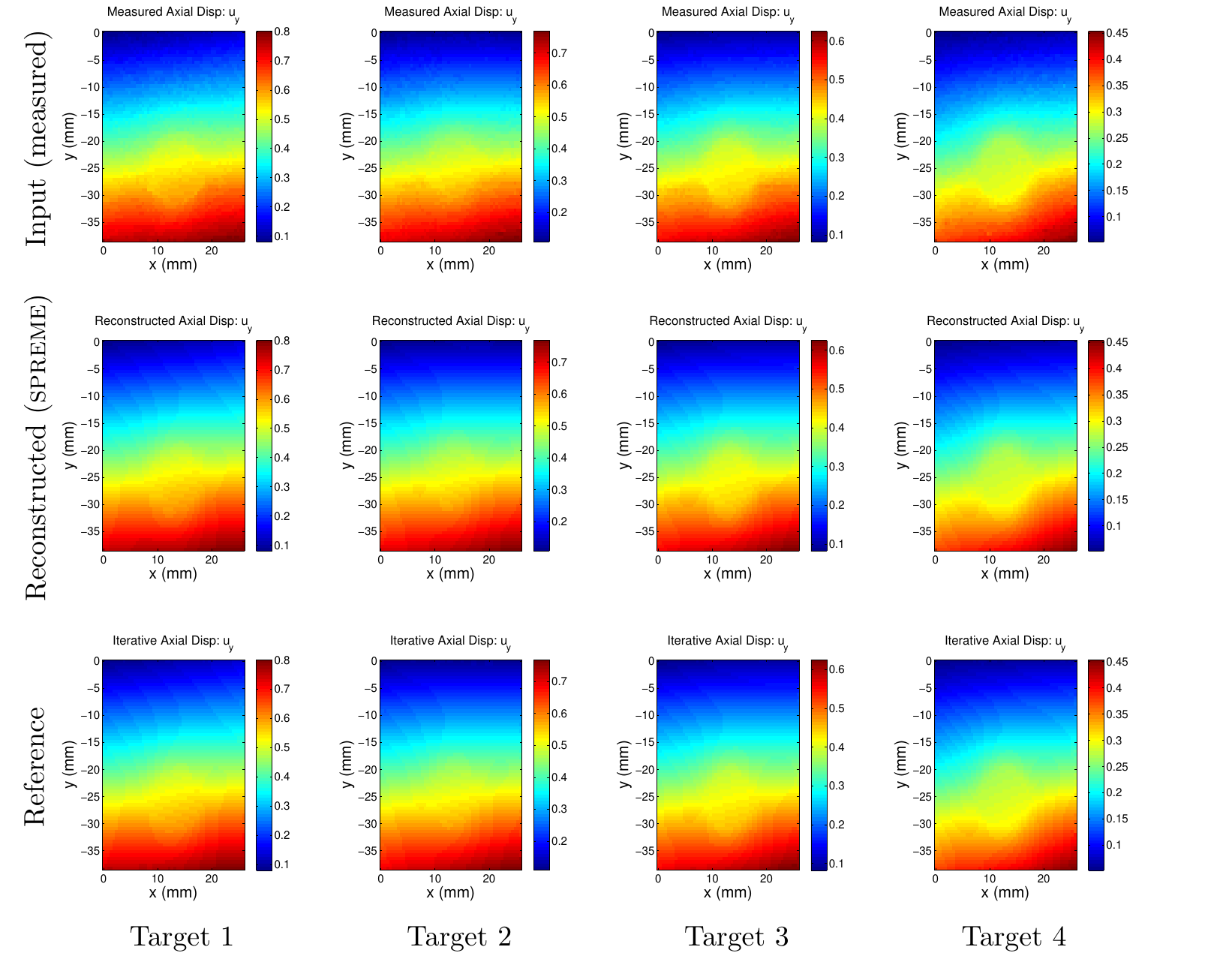}
  \caption{Axial displacement fields. First row: measured displacement field, second row: displacements reconstructed by \spreme,\ third row: reference displacements from an iterative inversion algorithm.}
  \label{fig:axDisp}
\end{figure}

\begin{figure}[!t]
\centering
          \includegraphics*{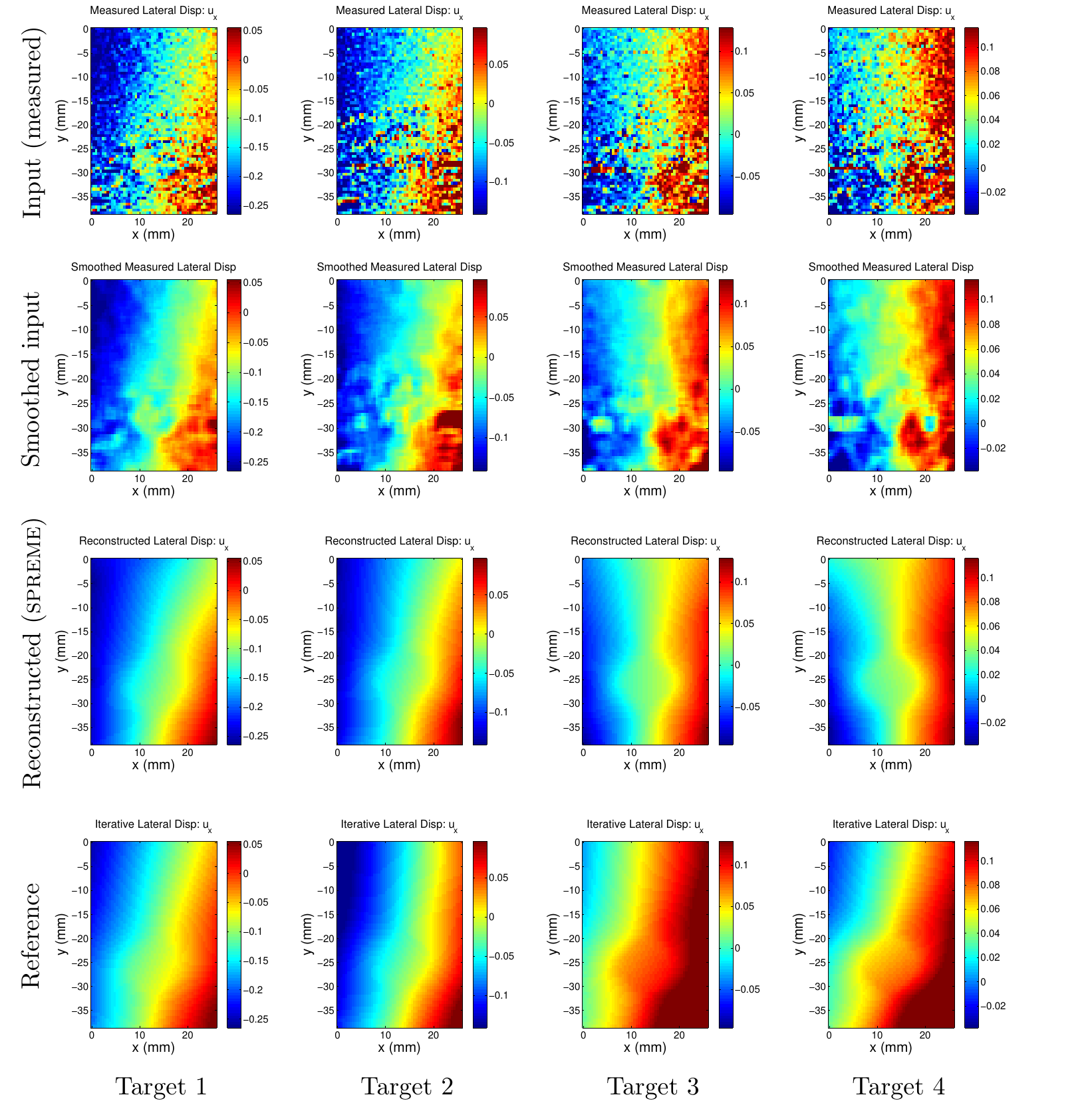}
  \caption{Lateral displacement fields. First row: measured displacements; second row: a smoothed version of the measured displacements; 
third row: displacements reconstructed by \spreme;\ fourth row: reference results from an iterative inversion algorithm. 
We note general agreement in gross features between all rows for targets 1 and 2. For targets 3 and 4, 
there is disagreement between the reference and the other fields, which we attribute to inaccuracies in 
assumed boundary conditions used to create the reference solution.   
}
  \label{fig:latDisp}
\end{figure}
 
The displacement for each target was downsampled to a 63 $\times$ 54 grid, 
which is the same grid that was used to generate the reference results \cite{jDordLNLinear}. 
Images of the measured, smoothed, reconstructed, and reference displacements are shown in figures (\ref{fig:axDisp}) and (\ref{fig:latDisp}).
The parameters used for \spreme\ processing were: $T_{xx} = 10^{-5},\ T_{yy} = 1,\ \alpha = 10^{-3},\ \beta = 10,\ \delta = 10^{-8},\ n = 0.5$.  Eleven iterations were performed for each target, though it was observed that the strains typically converged by the 5$^{\mbox{th}}$ iteration.

We note that the reconstructed axial displacements match the measured and reference axial displacements very closely. 
The reconstructed lateral displacements, however, were less similar to the reference displacements. This difference is most noticeable in targets 3 and 4. It is interesting to note that the general features in the \spreme\ displacement images for both these targets are more similar to the smoothed measurements than are the general features in the reference displacement images.
This indicates that the reference lateral displacements are in error, and this is likely due to errors in the assumed boundary conditions used to compute the reference results. 
\footnote{We recall that the iterative inversion code performs a forward elasticity solve at every iteration.  This solve requires assumed displacement and traction boundary conditions. The lateral direction is assumed to be traction free. The measured axial displacements were fixed as boundary conditions on all four edges of the phantom.}

\begin{figure}[!ht]
\centering
          \includegraphics*{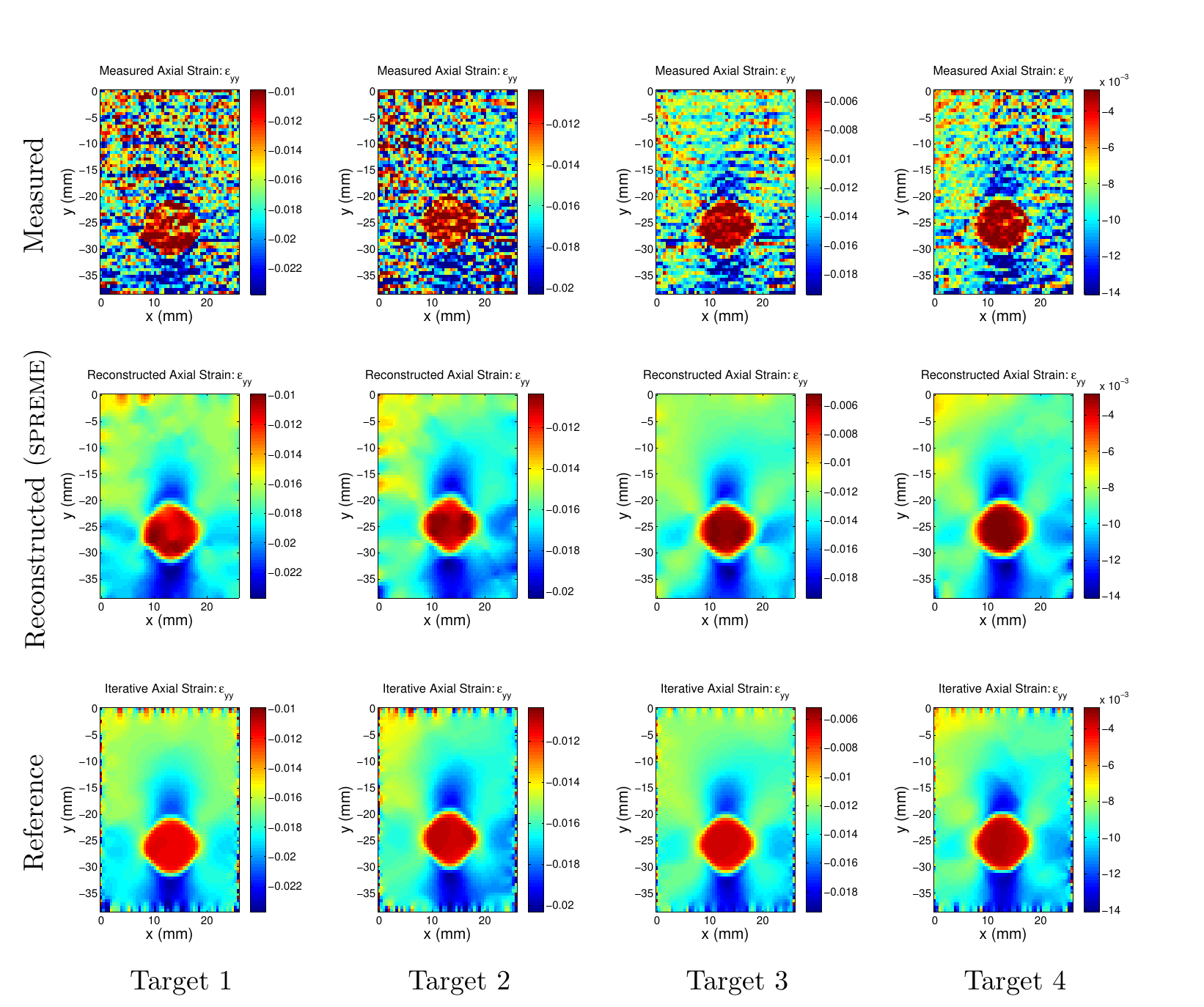}
  \caption{Axial strain fields. First row: measured strains; second row: strains reconstructed by \spreme;\ 
third row: reference results from an iterative inversion algorithm.
We note qualitative and quantitative agreement between the rows; we note also the 
presence of boundary artifacts in the reference fields that 
are not present in the \spreme\ fields.  
}
  \label{fig:axStrsP}
\end{figure}

\begin{figure}[!ht]
\centering
          \includegraphics*{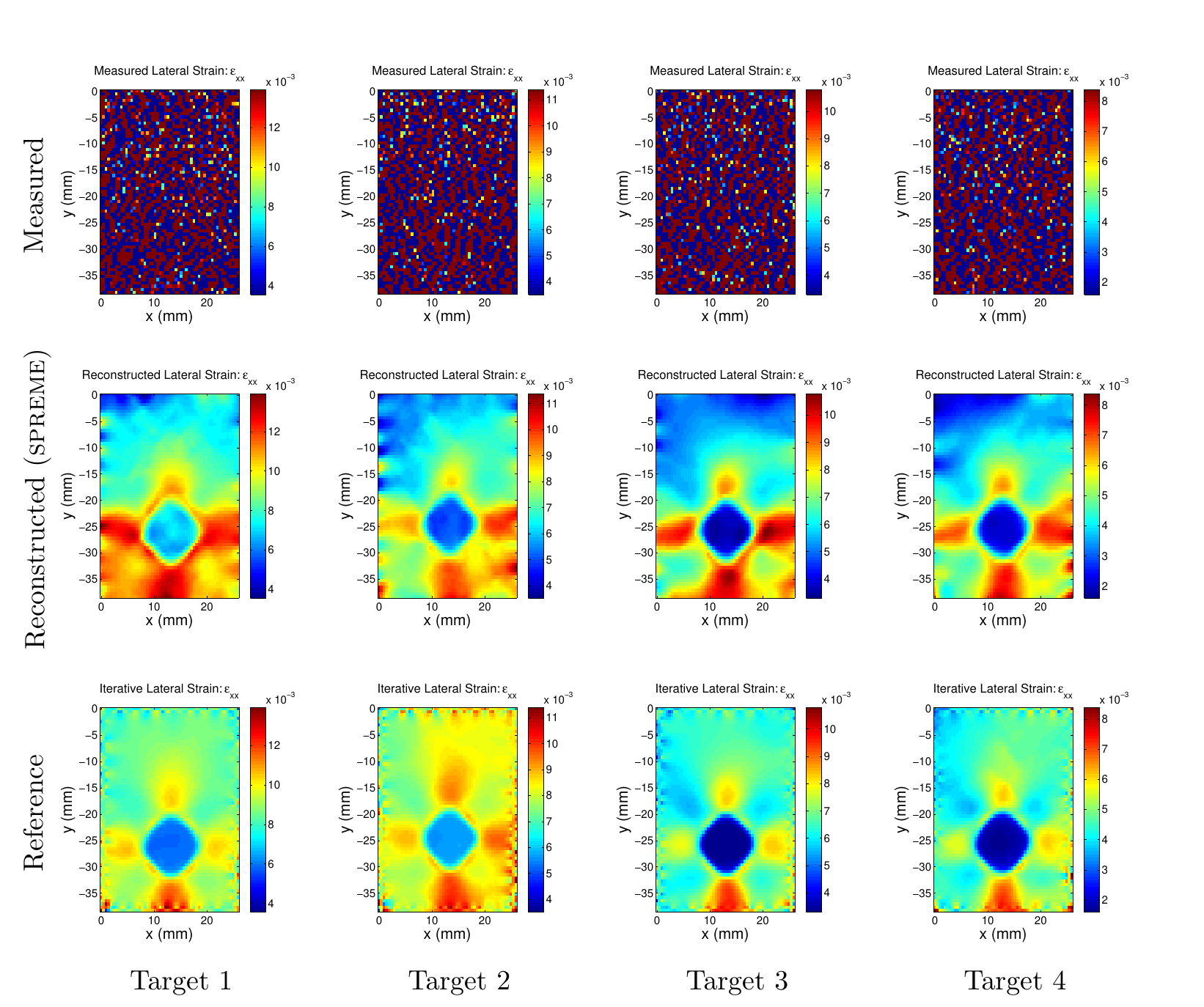}
  \caption{Lateral strain fields. First row: measured strains; second row: strains reconstructed by \spreme;\ 
third row: reference results from an iterative inversion algorithm.
We note general agreement of the main features between the \spreme\ and reference results; 
boundary noise present in reference results are not apparent in \spreme\ results;  
\spreme\ fields show higher dynamic range than 
reference fields.  
}
  \label{fig:latStrsP}
\end{figure}

Axial strain images ($\epsilon_{yy}$) 
for each target are shown in figure  (\ref{fig:axStrsP}), 
while 
the lateral strain fields ($\epsilon_{xx}$)  are in figure (\ref{fig:latStrsP}).
The measured (resp.\ reference) strains were computed by differentiating the 
measured (resp.\ reference) displacement fields using a finite difference approximation. 
Three observations may be made.  First, except for the lateral measured strain, 
all strain fields show clearly the presence of a stiff inclusion of the same size and shape, 
and with similar strain contrasts.  
The high noise level in the measured lateral displacement fields, 
however, is magnified through differentiation to a point where it 
obscures all features that might be present in the data. 
Second, we note boundary 
artifacts in the reference strain fields;  these are caused by the assumed boundary conditions in 
the reference reconstruction method.  The \spreme\ strain fields are relatively free of such artifacts.  Third, strain contrast or dynamic range is slightly (resp.\ significantly) higher in the \spreme\ axial (resp.\ lateral) strain reconstructions than in the reference fields.  This is perhaps due to deleterious effects of regularization used in the reference method, which tends to diminish modulus, and thus reconstructed strain, contrast.  Similar observations may be made when comparing shear strains, figure (\ref{fig:shearStrs}). 

\begin{figure}[!ht]
\centering
          \includegraphics*{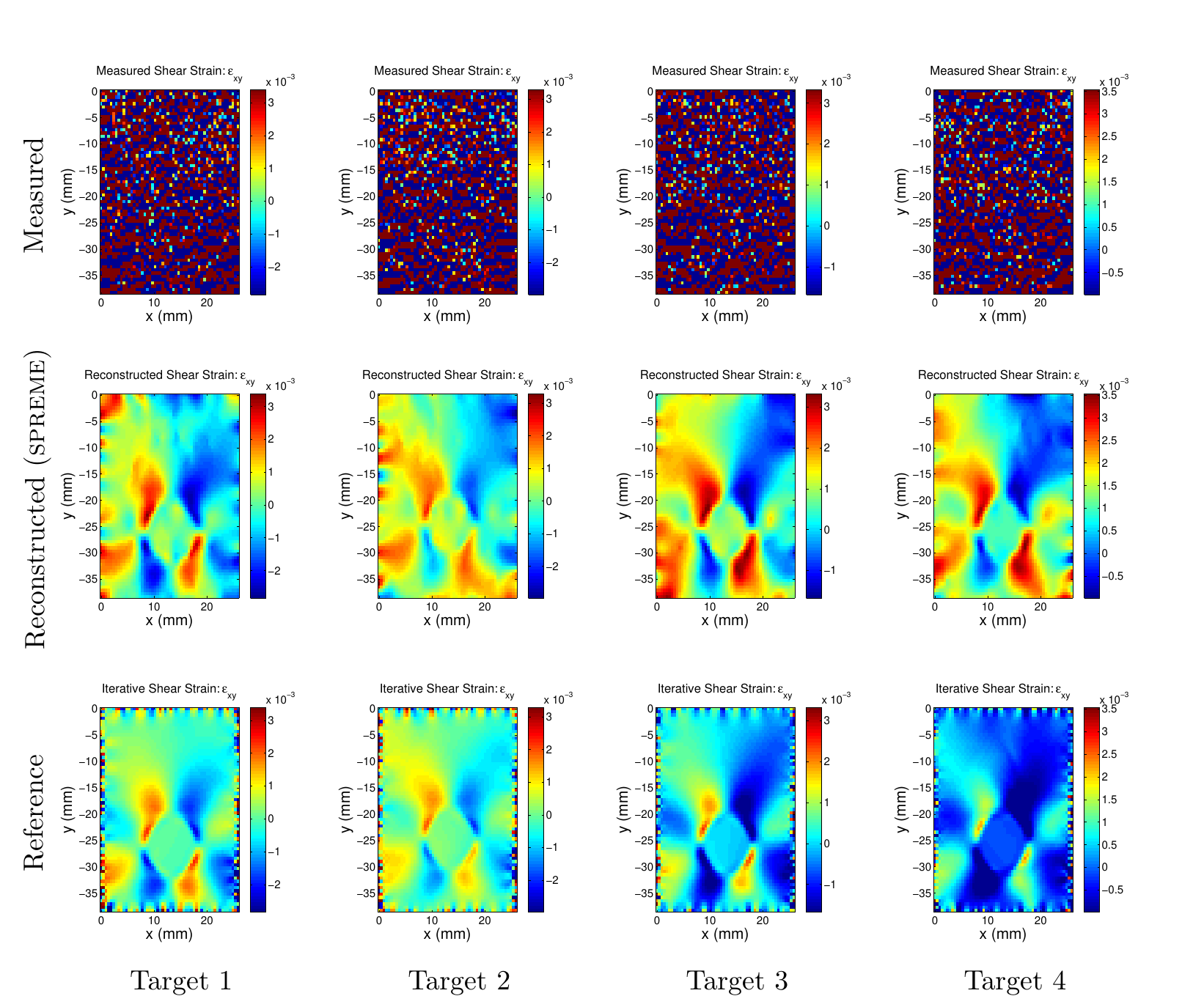}
  \caption{Shear strain fields. First row: measured strains; second row: strains reconstructed by \spreme;\ 
third row: reference results from an iterative inversion algorithm.
We note general agreement of the main features between the \spreme\ and reference results; 
boundary noise present in reference results are not apparent in \spreme\ results;  
\spreme\ fields show higher dynamic range than reference fields.  
}
  \label{fig:shearStrs}
\end{figure}
\

\section{Application to in-vivo image data}
\label{sec:app}
In this section, we test \spreme\ with in-vivo displacement data measured from patients with breast masses. The main goal of this test is to evaluate whether \spreme\ is sufficiently robust to provide useful results with clinical data, which is the motivating application. 
The in-vivo displacement data used for the tests has been processed with an iterative inversion code in a different study described in \cite{goenezen2012linear}. The results from that study will be used as a reference to benchmark the results produced by \spreme. 
%
%
These results contain displacement fields from $5$ biopsy proven fibroadenomas
and $5$ biopsy proven invasive ductal carcinomas, which 
represent the most common forms of benign and malignant breast tumors.
For the present 
purposes of demonstration, in this section we present results from one of each;
the reconstructed displacement and strains for the other eight cases are shown in appendix~(\ref{app:cDataImgs}).  For all the cases treated, 
we work on the same mesh as used in  \cite{goenezen2012linear}, 
and focus on displacement fields 
corresponding approximately to roughly 1\% overall strain in order to stay in the linear range.  
The input measured displacement fields used are shown in Figure (\ref{fig:measDisp_c}). 

%
%
%

\subsection{In-vivo results and discussion}
\begin{figure}[!ht]
\centering
%
          \includegraphics*{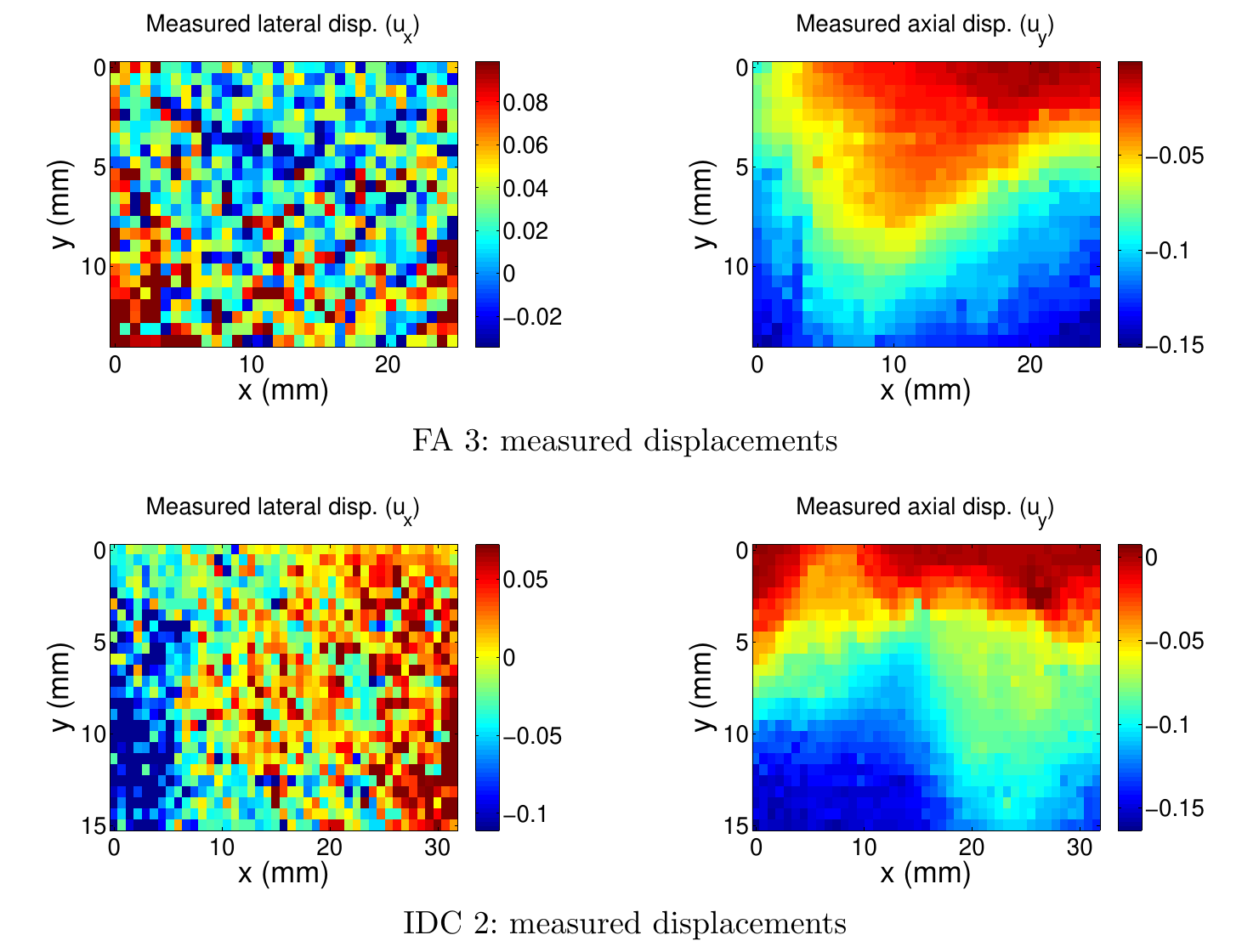}
  \caption{Displacement fields measured in-vivo in the breast as described in \cite{goenezen2012linear}. 
Upper row:   Fibroadenoma, a benign breast tumor (FA 3) 
Lower row:   Invasive Ductal Carcinoma, a malignant breast tumor (IDC 2). 
}
  \label{fig:measDisp_c}
\end{figure}


\begin{figure}[!ht]
\centering
          \includegraphics*{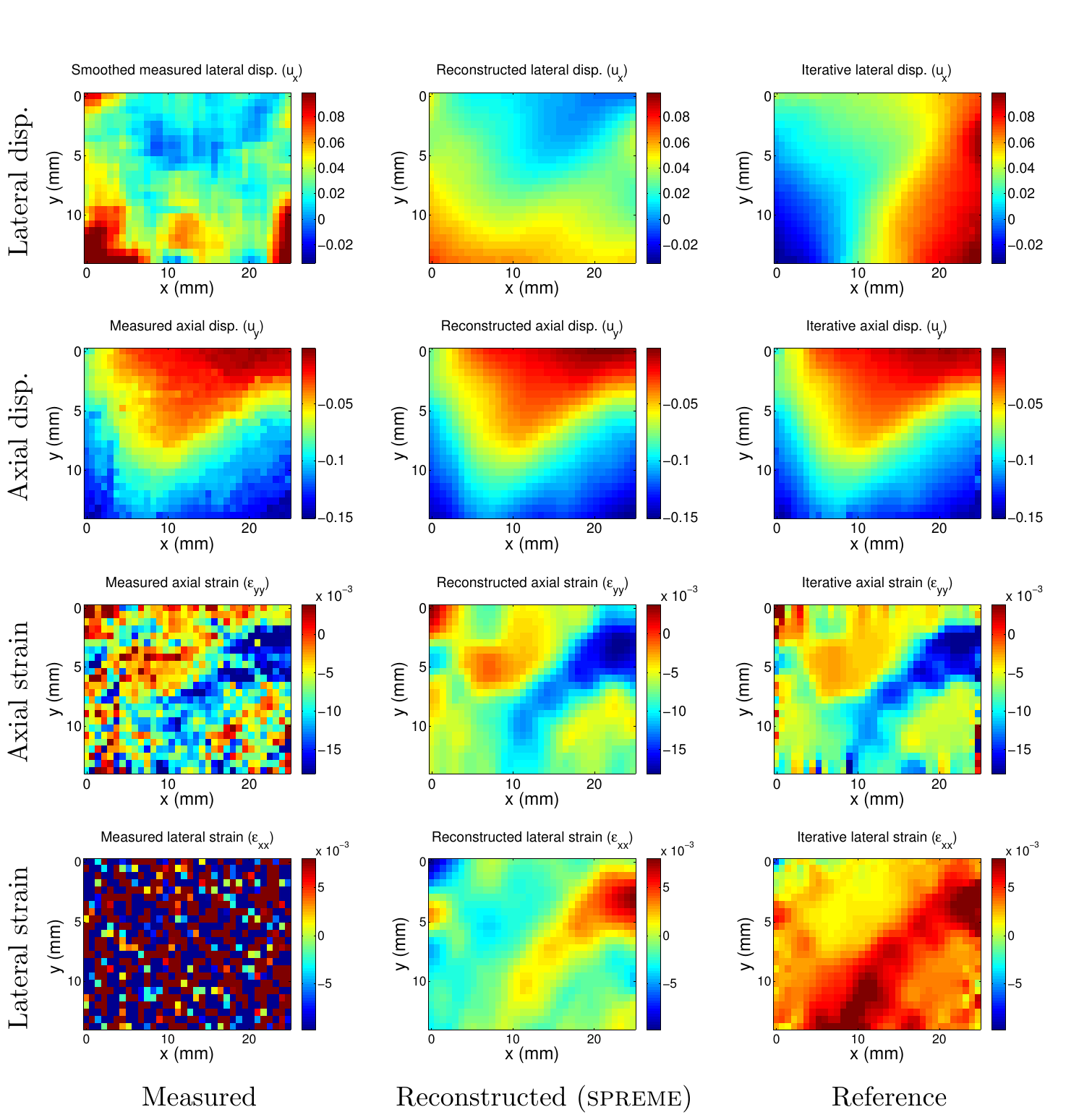}
  \caption{
Displacement and strain images for fibroadenoma 3.  We note that the reconstructed \spreme\ lateral displacements reproduce 
the gross features seen in the smoothed measurements, while the reference method does not.  
}
  \label{fig:fa3Plots}
\end{figure}

The parameters used for \spreme\ processing of the clinical data were:  $T_{xx} = 10^{-3},\ T_{yy} = 1,\ \alpha = 5 \times 10^{-4},\ \beta = 10,\ \delta = 10^{-8},\ n = 0.5$. 
The reconstructed displacement and strain fields obtained after the eleventh iteration are shown in figures (\ref{fig:fa3Plots}) and (\ref{fig:idc2Plots}), along with the reference fields.  
As with the validation study, we compare to two reference fields:  
the smoothed measured lateral displacement fields, 
and the results obtained by an
iterative inversion algorithm published in \cite{goenezen2012linear}. 
Generally speaking, we see excellent agreement between all three fields in the axial displacement and strain components.  

In the lateral displacement and strain fields, however, there is less 
consistency.   The \spreme\ lateral displacement fields tend to reproduce the gross features seen in the smoothed lateral displacement measurements;  the reference results, however, sometimes agree (c.f.\ figure (\ref{fig:idc2Plots})), and sometimes disagree (c.f.\ figure ((\ref{fig:fa3Plots})) with the other two.  
The reference displacements seem to indicate an outward motion at the left and right side of the imaged domain in all the tumor cases (even in the 8 other cases shown in the appendix). 
The smoothed measurements and the predicted displacements from \spreme,\ however, do not indicate this type of motion in all the tumor cases. This outward motion predicted by an iterative inversion code is probably an artifact due to the assumed boundary condition of traction free sides on the left and right sides of the imaged domain, and an implied symmetry.  These assumptions 
may not be always valid in practice.  

In contrast to the lateral displacement fields, the lateral strain fields recovered from \spreme\ 
were remarkably similar to the reference strains in all cases. This is surprising because the reconstructed lateral displacements were not always similar to the reference lateral displacements. This implies that the dominant difference between the \spreme\ and reference lateral displacement may be a rigid body mode.

\begin{figure}[!ht]
\centering
          \includegraphics*{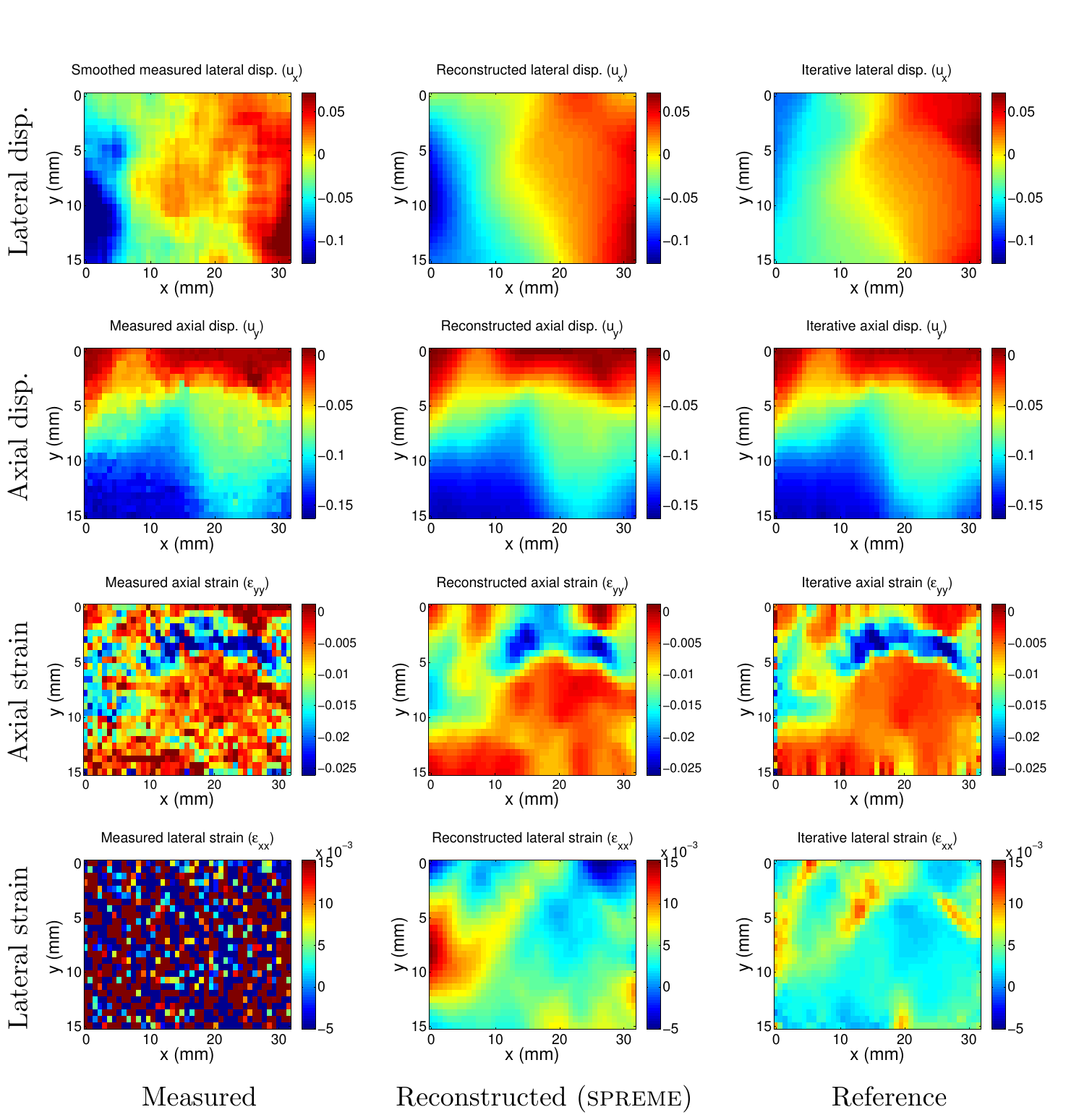}
  \caption{Displacement and strain images for invasive ductal carcinoma 2. 
In this example, we see more consensus between the gross features of 
the smoothed measurements, 
the reconstructed \spreme\ lateral displacements, 
and the reference lateral displacements.  
}
  \label{fig:idc2Plots}
\end{figure}

\section{Summary and conclusions}
\label{sec:summConc}
In this paper we consider the problem of reconstructing both components of a $2D$ vector displacement field in 
a heterogenous elastic medium, given a precise measurement of one of the two components and a noisy measurement of the other.  
This problem is motivated by applications in quasistatic ultrasound elastography, in which technological limitations allow the very precise measurement of one component of displacement, 
but impede the measurement of others.

Under the assumption that the material is piecewise homogenous, we proved 
that the momentum equations and knowledge of one component of the displacement field determines the other displacement component up to four undetermined coefficients.  
We then went on to derive a variational formulation that can exploit the condition of piecewise homogeneity without other {\em a priori} knowledge of modulus distribution.  

We then showed verification, validation, and application of 
a finite element implementation from this variational formulation.  The verification was performed 
on simulated data and showed that the method described here converges quickly, and dramatically reduces error in the noisy (uncertain) displacement component.  This came at the expense of a slight increase in noise in the precise components.  Nevertheless, in the verification studies conducted, the overall strain error reduced from roughly $O(1)$ (several hundred percent) to $O(10^{-1}-10^{-2})$ (a few percent).  
The validation studies showed that this method produces realistic displacement fields from 
measurements in tissue mimicking phantoms.  The predicted displacements tend to be more consistent with measurements than those of a reference method in the literature.  
Finally, the application of \spreme\ to data collected in-vivo demonstrates that \spreme\ is 
sufficiently robust to work in the application domain that originally motivated its development.   

Further improvements in lateral displacement estimation with \spreme\ are possible.  
Since \spreme\ is a post-processing method, it 
can be used in conjunction with the other lateral displacement estimation techniques outlined in the introduction to obtain more precise estimates of the displacement field. Since these methods produce improved measurements of the lateral displacements, then relatively large weights can potentially be used on the lateral displacements when processing them with \spreme.\ This processing can be done efficiently because \spreme\ converges in very few iterations.

An area that needs further work is finding a systematic method to identify good choices of 
the algorithmic parameters in \spreme.  Regularization parameter selection is an open and active area of research in inverse problems, and is not addressed in the present 
study.  

The \spreme\ approach may be generalized to other applications, as described in 
\cite{babaniyi:msthesis}. 
The idea is that one identifies the form of the momentum equation 
that results for a homogeneous material property distribution.  This equation is then enforced almost everywhere, as indicated in equation (\ref{eq:f3}).   See \cite{babaniyi:msthesis} for details.

\section*{Acknowledgments}
The authors will like to thank Prof.\ Timothy J.\ Hall and members of his lab at the University of Wisconsin for sharing the phantom and in-vivo displacement data.  The support of NIH Grant No. NCI-R01CA140271, and NSF Grant No. 1148124, and 1148111 is gratefully acknowledged.

\appendix
\section{Formulation without compatibility}
\label{app:1stFormu}
Given $\u_{m}(x,y)$, $\x\in\Omega$ and $\bep^{^{(k-1)}}$, find $\bep^{{k}}$ that minimizes:
\beq
\pi[\bep^k] = \pi_{_O} + \pi_{_R} \label{eq:ap1}
\eeq
where:
\beq
\pi_{_O} = \frac{1}{2} \int_{\Omega} \left( \nabla\u_{m} + (\nabla\u_{m})^{\mbox{T}} - 2\bep^k \right) \cdot \T \left( \nabla\u_{m} + (\nabla\u_{m})^{\mbox{T}} - 2\bep^k \right) \, d{\Omega} \label{eq:ap2}
\eeq
and:
\beq
\pi_{_{R}} = \frac{\alpha_{o}}{2}\int_{\Omega} \frac{\left(\nabla \cdot \A(\bep^{{k}})\right)^{2}}{\Big[\left(\nabla \cdot \A(\bep^{^{(k-1)}})\right)^{2} + \delta\Big]^n} \, d{\Omega}. \label{eq:ap3}
\eeq

In equation (\ref{eq:ap2}), $\T$ is a fourth order weighting tensor that allows more importance to be placed on the more accurate strain component. So in ultrasound measurements where the axial strains ($\ep_{yy}$) are typically more accurate than the lateral and shear strains ($\ep_{xx}, \ep_{xy}$), then the $T_{yyyy}$ weights  will be larger than $T_{xxxx}$ and $T_{xyxy}$. 

Equation (\ref{eq:ap1}) was minimized with respect to $\ep$ to get the weak form using the same method outlined in section \ref{ssec:VarFomu}. The weak form was then discretized with finite element bilinear shape functions to get a linear system of equations. A finite element within an in-house FEM code was created to solve the linear system of equations. 
\begin{figure}[!ht]
\centering
          \includegraphics*{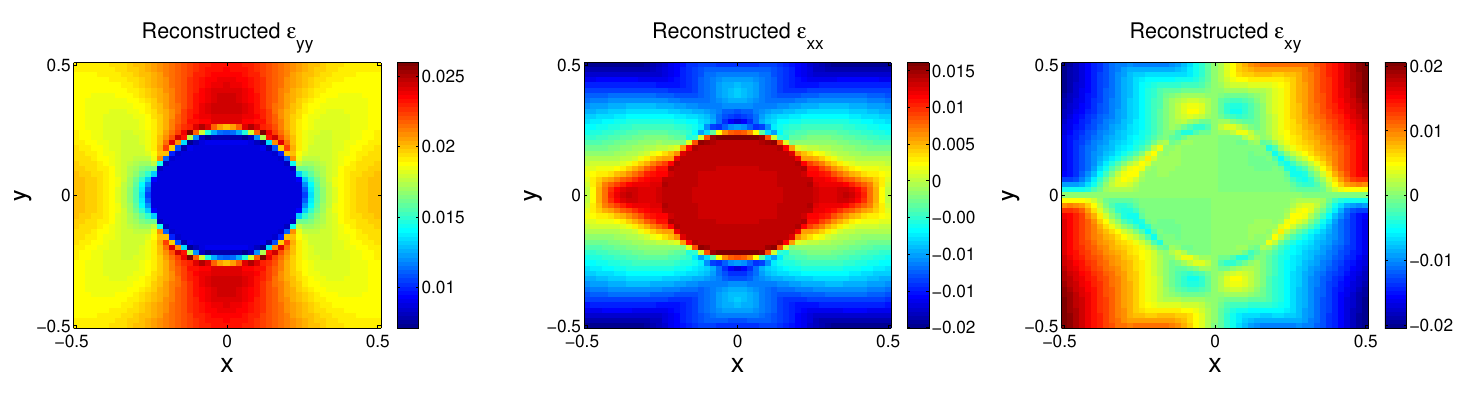}
  \caption{Reconstructed strains at fifth iteration.}
  \label{fig:recStrNoc}
\end{figure}

This formulation was tested with the perfect input displacements shown in figure (\ref{fig:targDisp}). The following parameters were used to run the code: $T_{yyyy} = 1$, $T_{xxxx} = T_{xyxy} = 1\times10^{-9}$, $\alpha = 1\times10^{-6}$, $\delta = 1\times10^{-8}$, $n=0.5$. These weights were chosen in order to simulate ultrasound measurements where the lateral and shear strains are not well known. Five iterations were performed, and the reconstructions obtained at the fifth iteration are shown in figure (\ref{fig:recStrNoc}). The lateral and shear strains look very different from the target distributions shown in figure 
(\ref{fig:test1strain}).
\begin{table}[!ht]
\begin{center}
        \begin{tabular}{ | p{3.0cm} | c | }
        \hline
        \centering Field variable & Error in reconstructed field (\%) \\ \hline   
        \centering $\ep_{xx}$ & $125$ \\ \hline
        \centering $\ep_{yy}$ & $1.41$ \\ \hline
        \centering $\ep_{xy}$ & $695$ \\ \hline
        \end{tabular}
\end{center}
\caption{Errors in the recovered strains after the fifth iteration.}
\label{tab:errStrVals}
\end{table}

To quantitatively determine how accurate the reconstructions were, the errors between the reconstructed strains at the 5th iteration and target strains were computed using the formula shown in equation (\ref{eq:eNorm_str}). The results of this computation are displayed in table (\ref{tab:errStrVals}). These results show that the lateral and shear strains are not accurate. 

Another method used to check the output strains from the FE code is the compatibility equations which are shown below \cite{book:atkinFox}:
\beq
\ep_{xx,yy} + \ep_{yy,xx} = 2\ep_{xy,xy}. \label{eq:ap4}
\eeq

The equation used to check the compatibility condition is slightly different from the equation shown in (\ref{eq:ap4}). The derivation of the equation used is shown below:
\bea
2\ep_{xy} &=& u_{x,y} + u_{y,x} \\
2\ep_{xy} - u_{x,y} &=& u_{y,x} \\
(2\ep_{xy} - u_{x,y})_{,y} &=& u_{y,xy} \\
(2\ep_{xy} - u_{x,y})_{,y} - \epsilon_{yy,x}&=& \eta. \label{eq:ap5}
\eea

In (\ref{eq:ap5}), $\eta$ is introduced as a measure of incompatibility. Setting $\eta = 0$ in (\ref{eq:ap5}), and taking its derivative with respect to $x$, we see that it is the same as equation (\ref{eq:ap4}). Therefore by knowing $u_x, \epsilon_{xy}$, and $\epsilon_{yy}$, we can use equation (\ref{eq:ap5}) to check that the strains from the FE code satisfy the compatibility equation. $\eta$ was evaluated pointwise within the domain for each of the experiments performed. To get a sense of how well the compatibility equation is satisfied for each experiment, we calculate the L2 norm. The formula used to evaluate the L2 norm is defined below: 
\beq
||\eta|| = \sqrt{\int_{\Omega} \eta^2 \ d\Omega}
\eeq

The L2 norm of $\eta$ was computed for a large range of experiments where the weights were varied, and the parameters $\alpha$, $\delta$, and $n$ were fixed to $1\times10^{-6}$, $1\times10^{-8}$, and $0.5$, respectively. The results from these experiments are shown in table (\ref{tab:compRes}). The results on the table demonstrate that the compatibility equation becomes increasingly violated as $T_{xxxx}$ and $T_{xyxy}$ are reduced, and the level of violation is maximum when $T_{xxxx} = T_{xyxy} = 10^{-5}$. The mesh size for the experiment is $0.02$ therefore any value of $||\eta||$ smaller than this mesh size can be neglected as discretization error. When $T_{xxxx}$ and $T_{xyxy}$ becomes smaller than $10^{-3}$, then the compatibility equation becomes violated for the experiment. This means that this alternate formulation cannot be used to calculate correct strains when the weights for the lateral and shear strain components are too small. When the weights for some of the strain components are small, the FE code constrains the corresponding strains with small weights to satisfy the equilibrium equations, but the strains are not forced to come from a single valued continuous displacement field. This formulation will therefore not work on real measured data where we would wish to use small weights to exclude the strain components calculated from the imprecisely measured displacement component.

\begin{table}
\begin{center}
        \begin{tabular}{ | p{3.0cm} | p{3.0cm} | p{3.0cm}| c | }
        \hline
        \centering $T_{yyyy}$ & \centering $T_{xxxx}$ & \centering $T_{xyxy}$ & $||\eta||$ \\ \hline
        \centering 1 & \centering 1 & \centering 1 & 0.013 \\ \hline
        \centering 1 & \centering $10^{-1}$ & \centering $10^{-1}$ & 0.013 \\ \hline
        \centering 1 & \centering $10^{-2}$ & \centering $10^{-2}$ & 0.013 \\ \hline
        \centering 1 & \centering $10^{-3}$ & \centering $10^{-3}$ & 0.017 \\ \hline
        \centering 1 & \centering $10^{-4}$ & \centering $10^{-4}$ & 0.093 \\ \hline
        \centering 1 & \centering $10^{-5}$ & \centering $10^{-5}$ & 0.127 \\ \hline
        \centering 1 & \centering $10^{-9}$ & \centering $10^{-9}$ & 0.127 \\ \hline
        \centering 1 & \centering 0 & \centering 0 & 0.127 \\ \hline
        \end{tabular}
\end{center}
\caption{Compatibility results using various weights.}
\label{tab:compRes}
\end{table}

By way of comparison, the \spreme\ formulation yields $||\eta|| \approx 9 \times 10^{-3}$.

\section{Clinical data reconstructions}
\label{app:cDataImgs}
In this section, for completeness, we show the reconstructions obtained from processing the rest of the in-vivo data 
from \cite{goenezen2012linear}.  The observations made in earlier apply to these results as well.  
\begin{figure}[!ht]
\centering
          \includegraphics*{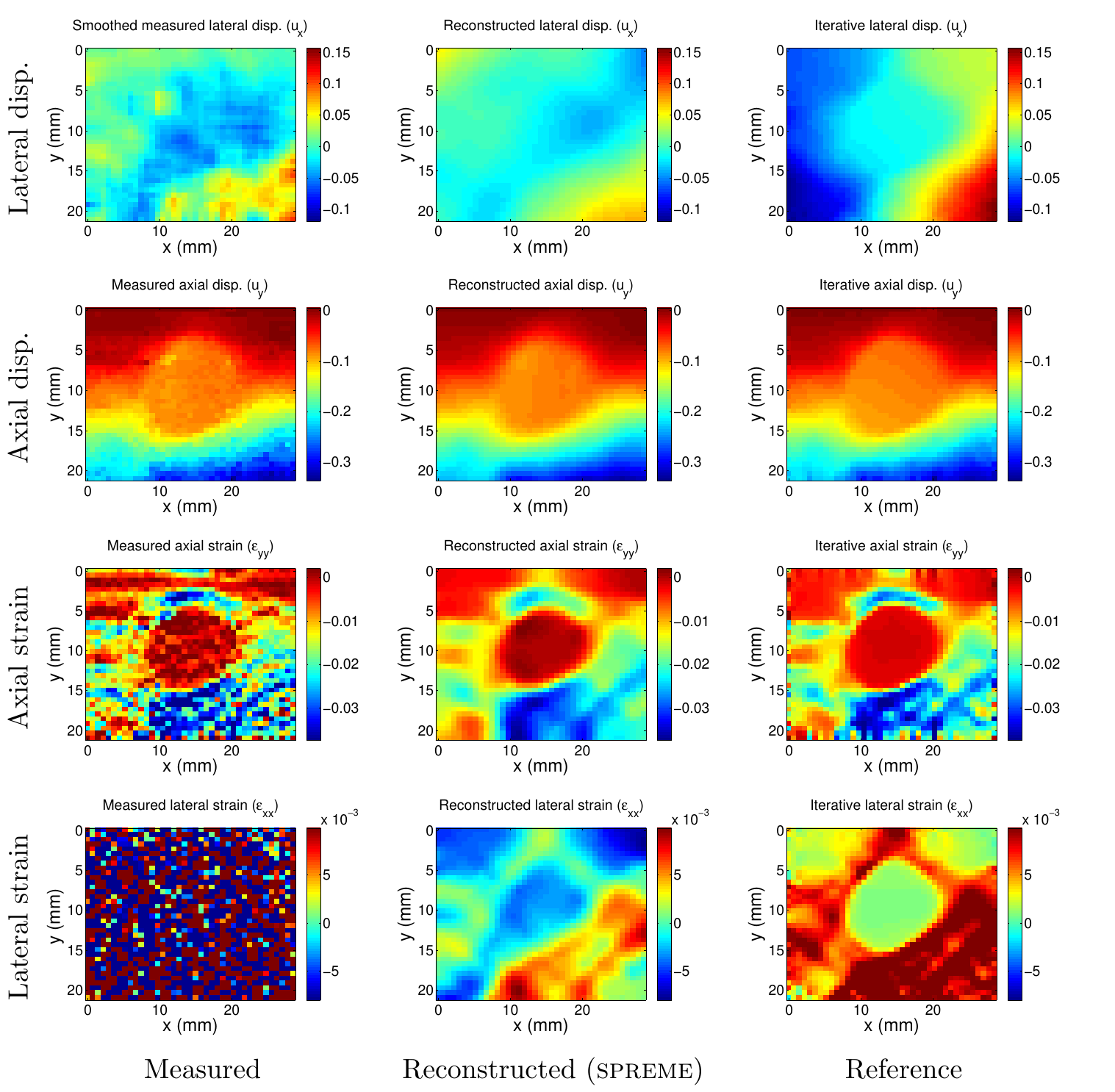}
  \caption{Displacement and strain images for fibroadenoma 1}
  \label{fig:fa1Plots}
\end{figure}

\begin{figure}[!ht]
\centering
          \includegraphics*{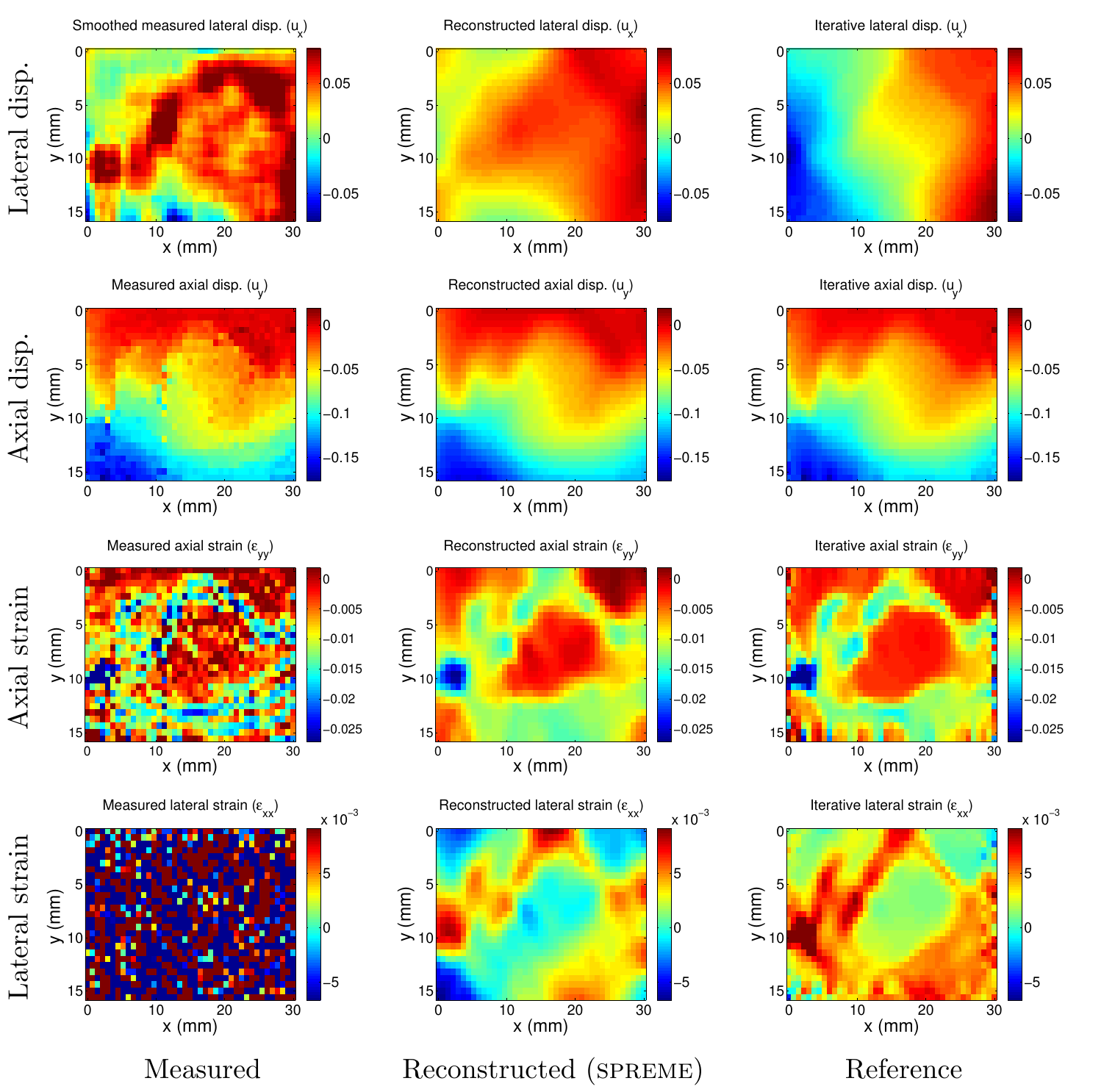}
  \caption{Displacement and strain images for fibroadenoma 2}
  \label{fig:fa2Plots}
\end{figure}

\begin{figure}[!ht]
\centering
          \includegraphics*{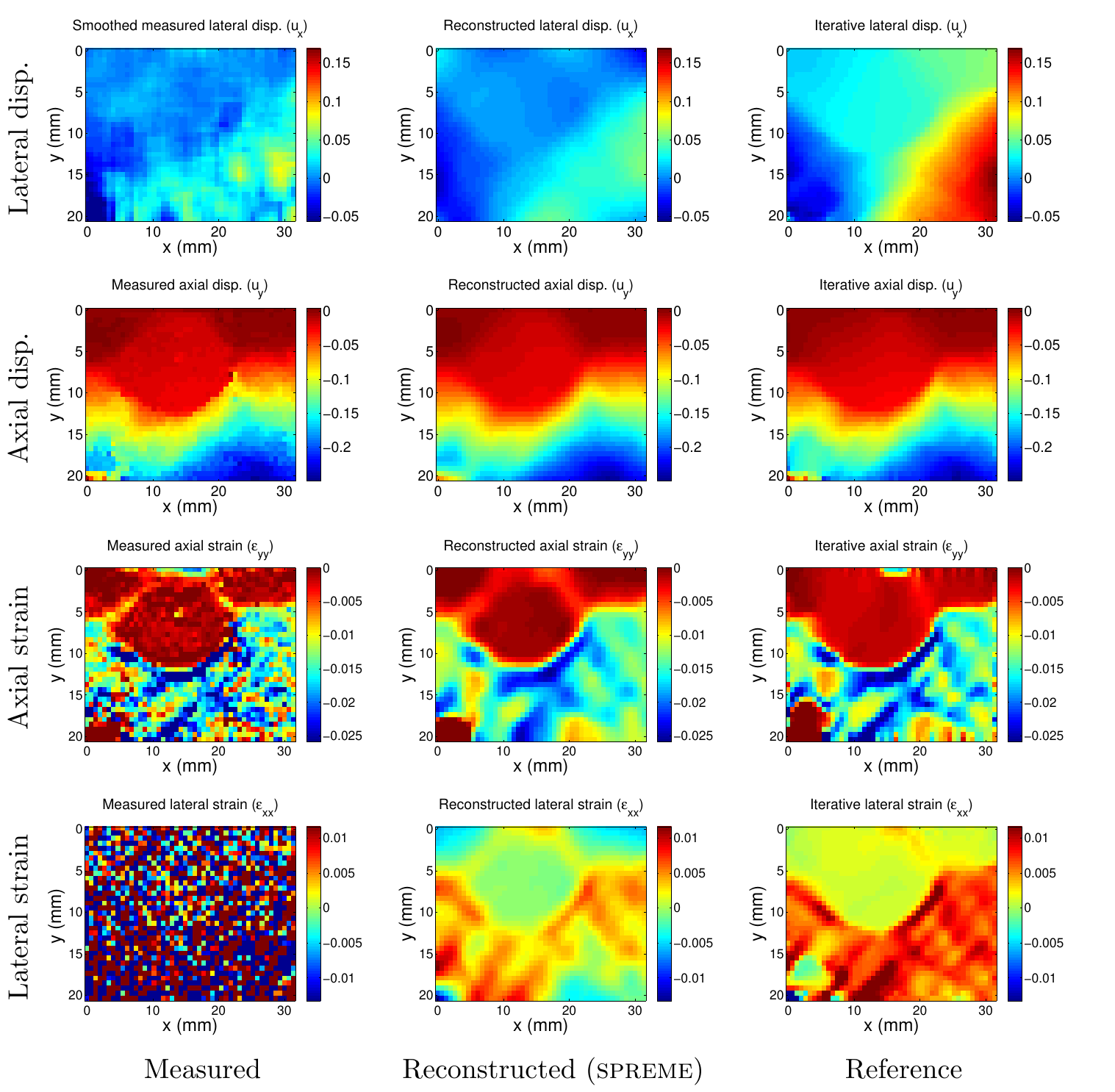}
  \caption{Displacement and strain images for fibroadenoma 4}
  \label{fig:fa4Plots}
\end{figure}

\begin{figure}[!ht]
\centering
          \includegraphics*{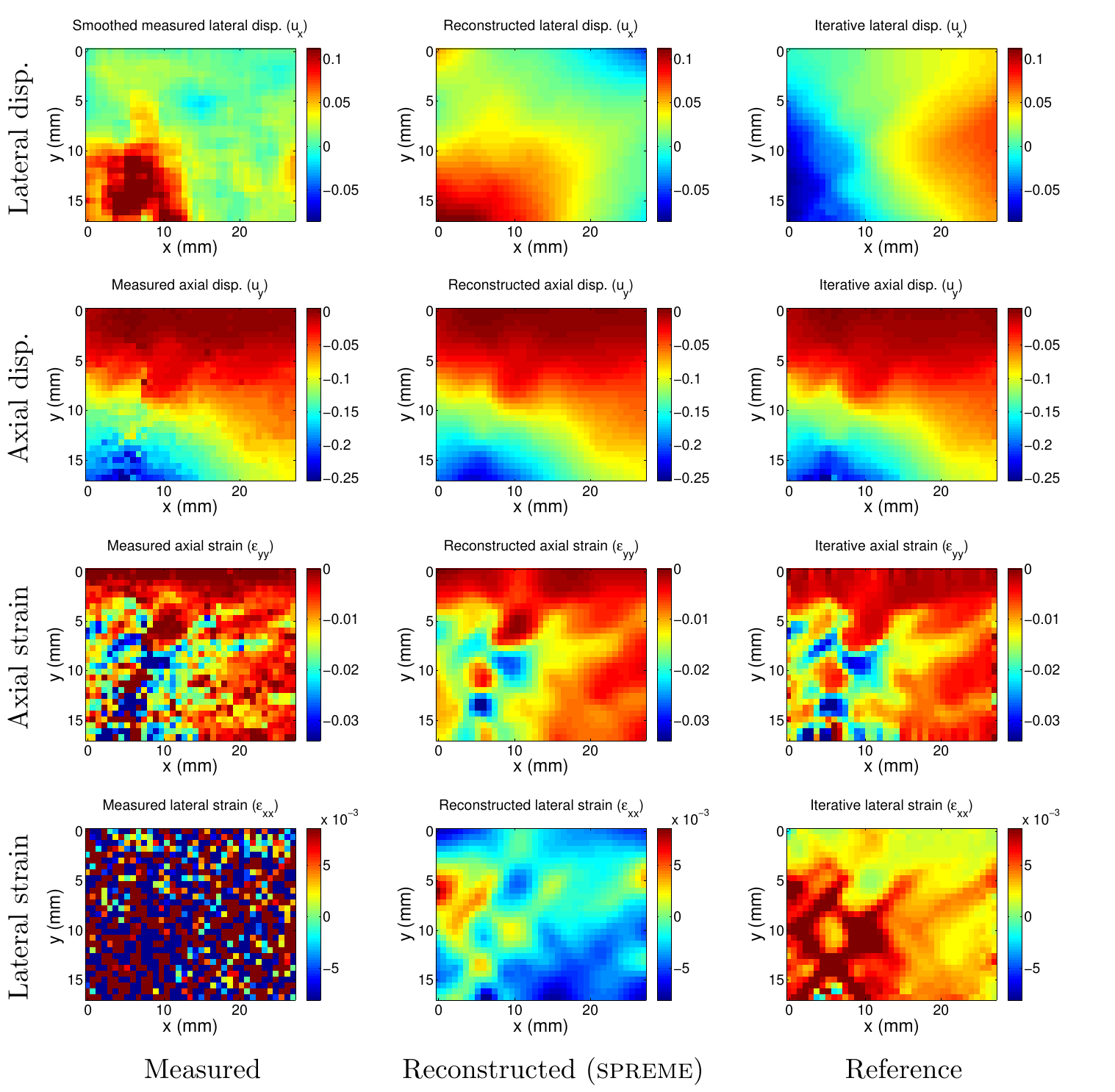}
  \caption{Displacement and strain images for fibroadenoma 5}
  \label{fig:fa5Plots}
\end{figure}

\begin{figure}[!ht]
\centering
          \includegraphics*{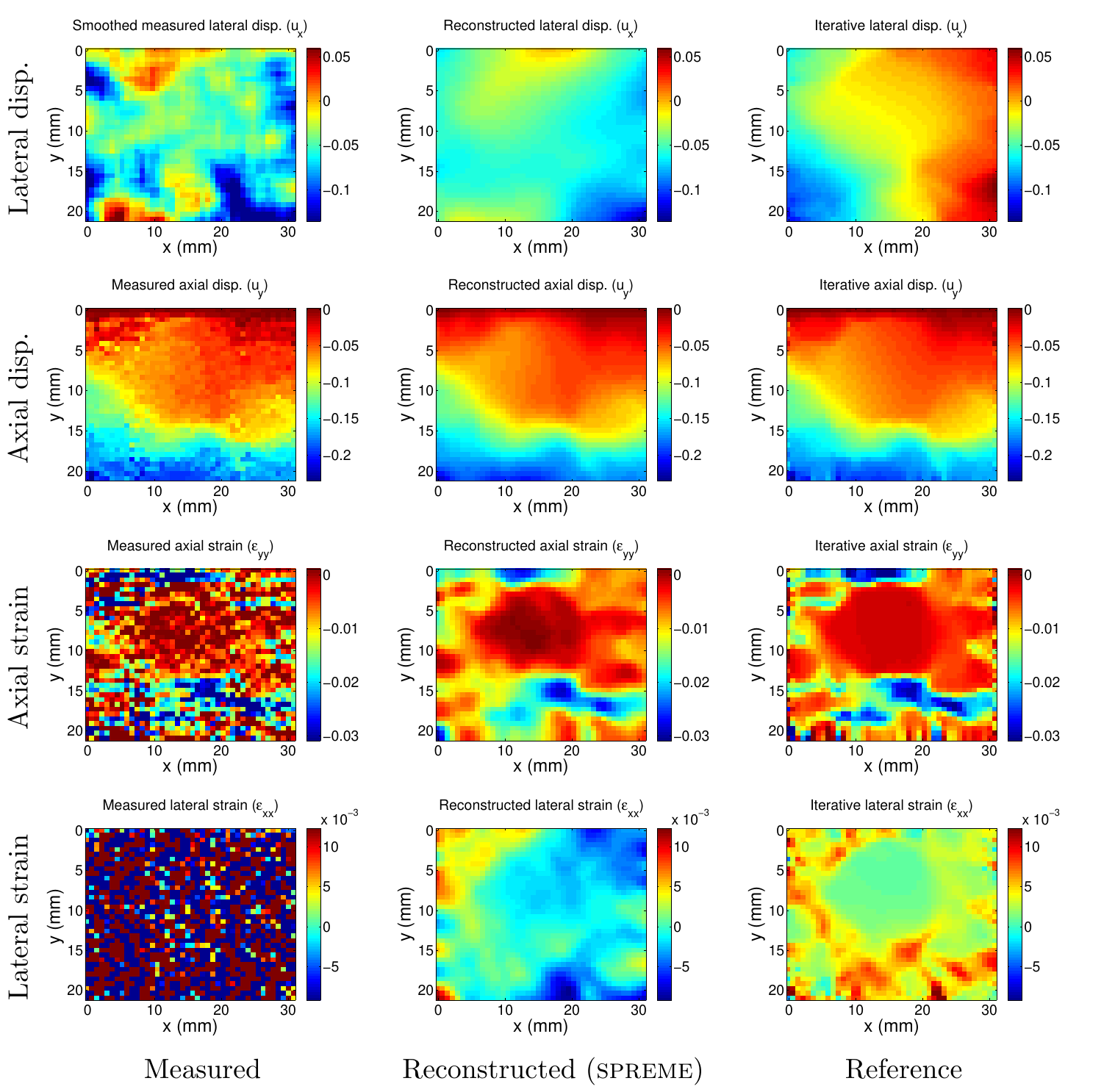}
  \caption{Displacement and strain images for invasive ductal carcinoma 1}
  \label{fig:idc1Plots}
\end{figure}

\begin{figure}[!ht]
\centering
          \includegraphics*{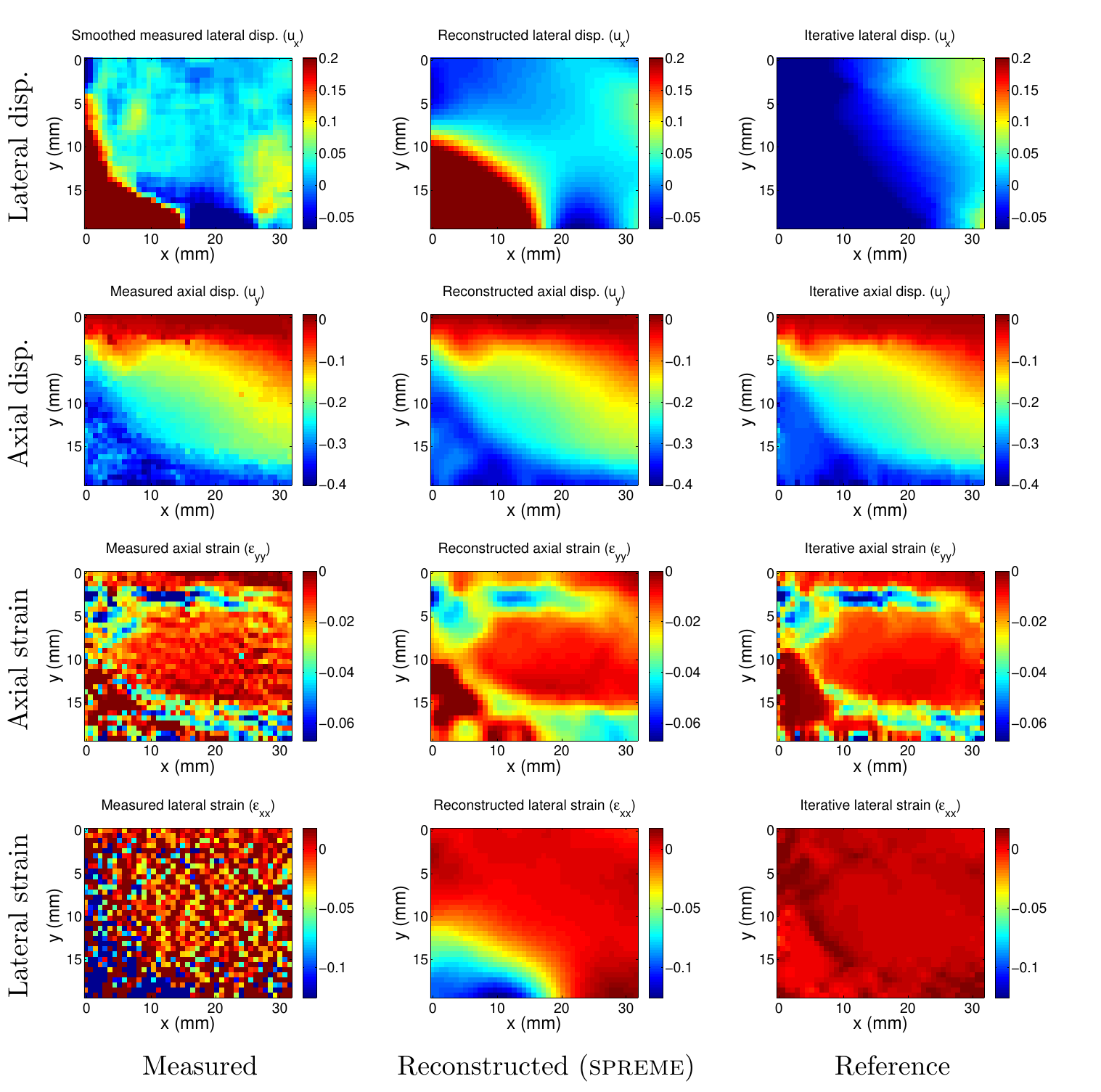}
  \caption{Displacement and strain images for invasive ductal carcinoma 3}
  \label{fig:idc3Plots}
\end{figure}

\begin{figure}[!ht]
\centering
          \includegraphics*{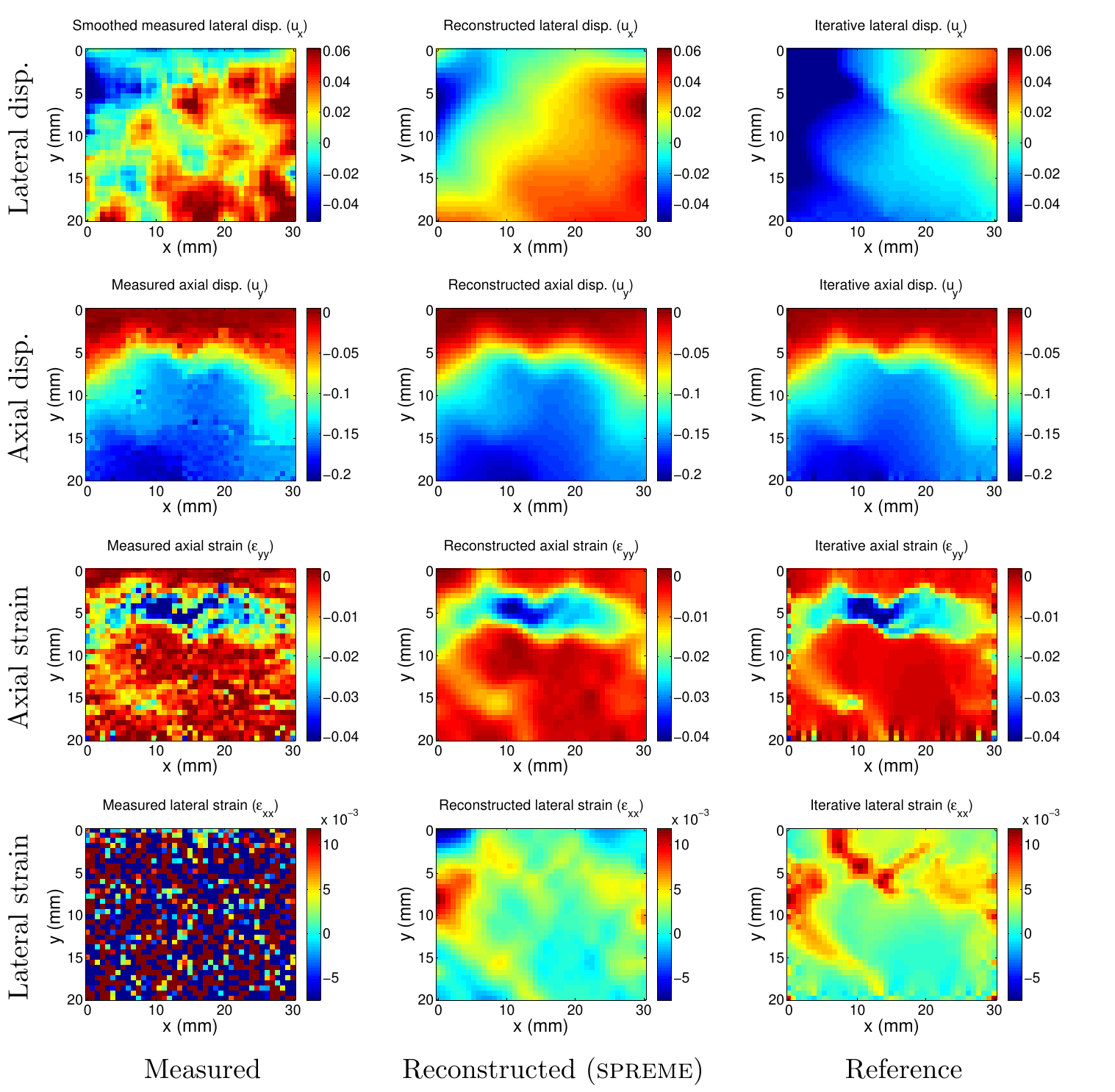}
  \caption{Displacement and strain images for invasive ductal carcinoma 4}
  \label{fig:idc4Plots}
\end{figure}

\begin{figure}[!ht]
\centering
          \includegraphics*{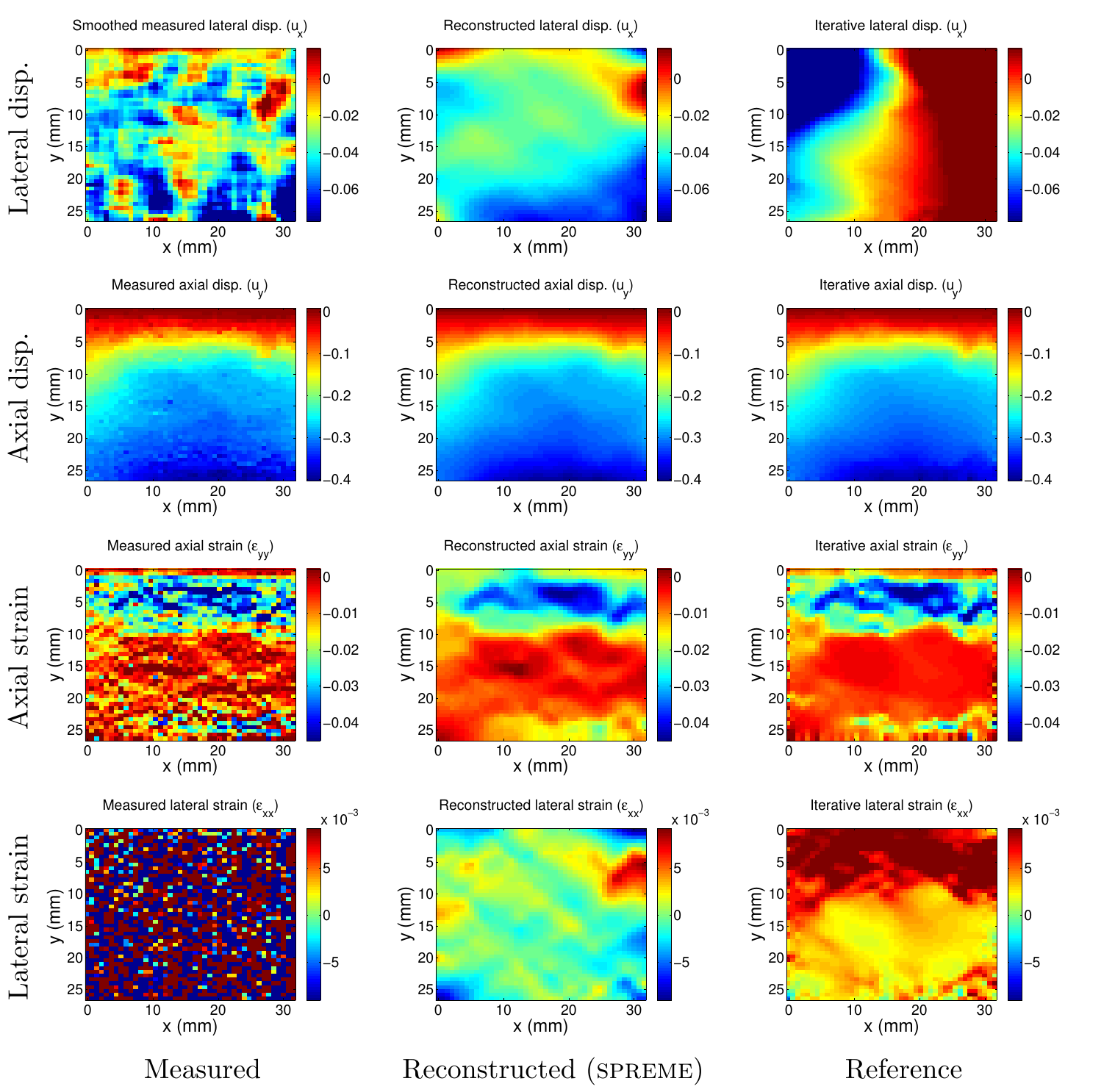}
  \caption{Displacement and strain images for invasive ductal carcinoma 5}
  \label{fig:idc5Plots}
\end{figure}

\clearpage
\bibliographystyle{apalike}
\bibliography{./dispProcess.bib}

\end{document}